\documentclass[article, nojss]{jss} 

\graphicspath{{figure/}{./}}
\usepackage{natbib}
\usepackage{mathtools, amssymb, amsthm, mathrsfs, amsmath}
\usepackage[normalem]{ulem}
\usepackage[utf8]{inputenc}
\usepackage[ruled,vlined]{algorithm2e}
\usepackage{placeins} 

\newcommand{\QV}{QV} 

\DeclareMathOperator*{\argmax}{arg\,max}
\DeclareMathOperator*{\argmin}{arg\,min}

\newtheoremstyle{thm}
  { }
  { }
  {\itshape}
  {4pt}
  {\bfseries}
  { }
  {\newline}
  {\thmname{#1}\thmnumber{ #2}:\thmnote{ #3}}
\theoremstyle{thm}
\newtheorem{theorem}{Theorem}[section]

\author{Andreas Nordland\\ Section of Biostatistics \\ University of Copenhagen
   \AND Klaus Kähler Holst \\ Novo Nordisk}
\Plainauthor{Andreas Nordland, Klaus Kähler Holst}

\title{Policy Learning with the \pkg{polle} package}
\Plaintitle{Policy Learning with the polle package}
\Shorttitle{Policy Learning (polle)}

\Abstract{
  
The \proglang{R} package \pkg{polle} is a unifying framework for learning and
evaluating finite stage policies based on observational data. The package
implements a collection of existing and novel methods for causal policy
learning including doubly robust restricted Q-learning, policy tree learning, and
outcome weighted learning. The package deals with (near) positivity violations
by only considering realistic policies. Highly flexible machine learning methods
can be used to estimate the nuisance components and valid inference for the
policy value is ensured via cross-fitting. The library is built up around a
simple syntax with four main functions \code{policy_data()},
\code{policy_def()}, \code{policy_learn()}, and \code{policy_eval()} used to
specify the data structure, define user-specified policies, specify policy
learning methods and evaluate (learned) policies. The functionality of the
package is illustrated via extensive reproducible examples.

}

\Plainkeywords{policy learning, dynamic treatment regimes, semi-parametric inference,
  double machine learning, R}

\Address{
  Andreas Nordland \\
  Section of Biostatistics\\
  University of Copenhagen \\
  Øster Farigmagsgade 5\\
  1014 Copenhagen\\
  Denmark \\
  E-mail: \email{andreasnordland@gmail.com}

  Klaus Kähler Holst \\
  Novo Nordisk\\
  Vandtårnsvej 108-110\\
  2860 Søborg\\
  Denmark \\
  E-mail: \email{kkzh@novonordisk.com}
}

\begin{document}

\section{Introduction} \label{sec:intro}

Sequential decision problems arise in various fields. Important examples
include deciding on treatment assignments in a medical application, defining
equipment maintenance strategies in a military or industrial setting, or
determining a sales strategy in a commercial context. Policy learning seeks to
identify sequential decision strategies from observational data and to quantify
the effect of implementing such a strategy using causal inference techniques.
While the theoretical field has progressed substantially during the last decade
based on advances in semiparametric methods, there has been a large gap in terms
of generic implementations of these methods being available to practitioners.

The \proglang{R} package \pkg{polle} \citep{cran:polle} is a unifying framework
for learning optimal policies/dynamic treatment regimes from historical data
based on cross-fitted doubly robust loss functions for finite horizon problems
with discrete action sets. Within this scope, the package unifies available
methods from other \proglang{R} packages and introduces previously
unavailable methods. The performance of the methods can then easily be
evaluated, compared and applied to new data. As a unique feature, the package
also handles a stochastic number of decision stages. In addition, the package
deals with issues related to learning optimal policies from observed data under
(near) positivity violations by considering \emph{realistic policies}.

The core concept of \pkg{polle} is to use doubly robust scores/double machine
learning developed from semiparametric theory when estimating the value of a
policy \citep{robins1986new, chernozhukov2018double}. These scores are also used
to construct a doubly robust loss function for the optimal policy value
\citep{tsiatis2019dynamic}. The resulting loss function is the basis for policy
value search within
a restricted class of policies such as policy trees \citep{athey2021policy}.
Transformations of the value loss function leads to a range of other loss
functions and methods such as doubly robust $Q$-learning
\citep{luedtke2016super} and outcome weighted learning based on support vector
machines \citep{zhang2012estimating, zhao2012estimating}. As is customary within
targeted learning and double machine learning, our policy evaluation and policy
learning methods apply cross-fitting schemes, which allow for inference under
weak conditions even when nuisance parameters are learned from highly flexible
machine learning methods.

Recursive policy learning is closely related to estimating heterogeneous
causal effects via the conditional average treatment effect, see
\citep{kennedy2020towards, semenova2021debiased} for recent overviews of the
field and some of the challenges regarding inference and generic error bounds.
Other notable mentions include \citep{kunzel2019metalearners} and
\citep{athey2019generalized}. Variable importance measures formulated as
projections are also closely related to policy learning, see
\citep{van2006statistical}.

The \pkg{polle} \proglang{R} package joins a collection of other packages
available on CRAN,
\url{https://CRAN.R-project.org/view=CausalInference}.
Other packages which should be highlighted include \pkg{DynTxRegime}
\citep{cran:dyntxregime} which provides
methods for estimating policies including interactive $Q$-learning, outcome
weighted learning, and value search. However, most of the methods are only
implemented for single stage problems and the package has no cross-fitting
methods for consistent policy evaluation. The \pkg{polle} package wraps
efficient augmentation and relaxation learning and residual weighted learning
from the \pkg{DynTxRegime} package.
The package \pkg{policytree} \citep{sverdrup2020policytree, policytree} is an
implementation of single stage policy tree value search based
on doubly robust scores. The \pkg{polle} package wraps this functionality and
extends it to a stochastic number of stages. A related \proglang{R} package is
\pkg{grf} \citep{grf2024}, which implements causal forests for the conditional
average treatment effect. Lastly, the \proglang{R} package \pkg{DTRlearn2} \citep{cran:DTRlearn2}
implements outcome weighted learning in a fixed number of stages.
The \pkg{polle} package also wraps this functionality.

Beyond \proglang{R}, the \proglang{Python} package \pkg{EconML} \citep{econml}
implements a wide range of learners for the conditional average treatment
effect including doubly robust estimators (equivalent to doubly robust
$Q$-learning as formulated in \pkg{polle}), double machine learning estimators,
and orthogonal random forests. The package also implements policy trees and
forests. To our knowledge, \pkg{EconML} does not contain methods for cross-fitted
policy evaluation. For multi-stage decision problems, the package only
considers G-estimation based on specific Markov decision process structural
equation models, see \citep{lewis2020double}.

The available methods for policy learning in proprietary software are still
very limited. A \proglang{SAS} macro denoted \code{PROC QLEARN} performs
standard $Q$-learning \citep{procqlearn}. In \proglang{stata} methods are
limited to estimating average treatment effects and potential outcome means
with the \code{teffects} function \citep{stata}.

In Section \ref{sec:concepts} we introduce the most important concepts of
doubly robust policy learning in a simple single-stage setting. In doing so, we
avoid the cumbersome notation needed for the general sequential setup as presented in
Section \ref{sec:methods}. In Section \ref{sec:syntax} we give an overview of
the package syntax and describe the main functions of the package.
Section \ref{sec:examples} contains four reproducible examples
based on simulated data covering all aspects of the package.
In Section \ref{sec:harvard_example} we present a complete analysis of a data set
investigating the treatment effect of a literacy intervention.
Finally, in Section \ref{sec:summary} we summarize the
functionality of \pkg{polle} and discuss limitations and future developments.

\section{Concepts} \label{sec:concepts}

In a randomized trial investigating the average treatment effect of two
competing treatments we should ask ourselves
whether the treatment effect is heterogeneous or not,
i.e., whether the subjects respond differently to the treatments depending on
their age, sex, disease history, etc. If so, is it possible to learn a treatment policy
from the observed data that will have a greater expected outcome than any of the
individual treatments?

For simplicity, we consider a single-stage problem, where each subject
receives a completely randomized treatment at a single time point. Let $A$ denote the
treatment variable with two levels $A\in \{0,1\}$, and let $U$ denote the measured
utility outcome. The average treatment effect is a causal parameter which
can be formulated via potential outcomes
\citep{rubin1974estimating, hernan2010causal}.
We let $U^a$ denote the potential utility had we forced the subject to receive
treatment $A=a$, and we refer to $\E[U^a]$ as the value of the given treatment.
The average treatment effect is now defined as the difference in value,
$\E[U^1 - U^0]$, and due to complete randomization, the effect is identified as
$\E[U|A=1] - \E[U|A=0]$. The effect is easily estimated based on a sample of
$N$ iid observations $O = (A, U)$. In the following we will assume that treatment
$A=1$ is recommended if the estimated effect is positive and vice versa.

\begin{figure}[htbp]
\centerline{\includegraphics[width=0.3\textwidth]{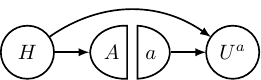}}
\caption{Single world intervention graph illustrating confounding via the history.}
  \label{fig:swig}
\end{figure}

Suppose now that we also collect a set of baseline covariates
$H\in \mathcal{H}$ for each subject, and that treatment randomization
depends on this history by design. The treatment probability model is then given
by a known function
\begin{align}
g_0(h, a) = \Prob(A=a|H=h). \label{eq:concept_pos}
\end{align}
If the trial has a sensible design we will also know that $g_0(H, a)>0$ almost
surely for $a\in\{0,1\}$, which is commonly referred to as the positivity condition.
Due to confounding, the average treatment effect will no longer be identified by the mean
utility in each treatment group,
see Figure \ref{fig:swig}.
However, it is possible to show that
\begin{align}
\E\left[\frac{I\{A=a\}}{g_0(H, a)} U\right] = \E[U^a]. \label{eq:concept_ipw}
\end{align}
This equality inspires an inverse probability weighting estimator for
the value of each treatment group. In an observational study, the treatment
probability function $g_0$ might not be known a priori. In that case, we instead
use some appropriate estimate $g_N$.

Alternatively, the treatment value is identified as
\begin{align}
\E[\E[U|A=a, H]] = \E[U^a], \label{eq:concept_Q}
\end{align}
where we define the quality function as $Q_0(h, a) = \E[U|A=a, H=h]$.
Usually, the $Q$-function is not known a priori and will need to be estimated.
The fit $Q_N$ is then used to construct an outcome regression estimate
of the value in each treatment group, see \citep{robins1986new}.\\

If the treatment effect is heterogeneous across the collected history it
is possible that one group of subjects benefit from treatment $A=1$ and that
another group of subjects benefit from treatment $A=0$. The researcher may even
have a candidate treatment policy
$d:\mathcal{H}\rightarrow \{0,1\}$ that is believed to improve the
value.
Let $U^d$ denote the potential utility had we forced the subject to be treated
in accordance to policy $d$. The policy value $E[U^d]$ can then be estimated
using \eqref{eq:concept_ipw} or \eqref{eq:concept_Q} or a combination of the two.
Define the doubly robust policy score as
\begin{align}
Z(d, g, Q)(O) = Q(H, d(H)) + \frac{I\{A=d(H)\}}{g(H, A)} \left(U - Q(H, A) \right). \label{eq:Z_d_K1}
\end{align}
If either $g = g_0$ or $Q = Q_0$ it holds that
\begin{align}
\E\left[Z(d,g, Q)(O)\right] = \E[U^d]. \label{eq:E_Z_d_K1}
\end{align}
The associated empirical mean plug-in estimator of the policy value
\begin{align*}
\theta_{N} =  N^{{-1}}\sum_{i=1}^{N} Z(d, g_{N}, Q_{N})(O_{i})
\end{align*}
is said to be doubly robust. This estimator is a central component of the
\pkg{polle} package and it is implemented in the function \code{policy_eval()}.

Furthermore, it is possible to show that the doubly robust policy value
estimator is
asymptotically efficient, if the nuisance models ($g_N$ and $Q_N$)
are correctly specified, see \citep{van2003unified}. Specifically,
the centralized score
\begin{align*}
Z(d, g, Q)(O) - \E\left[Z(d,g, Q)(O)\right]
\end{align*}
is the efficient influence function from which we can derive the asymptotic
distributions via central limit theorem arguments. Specifically,
\begin{align*}
N^{{-1}}\sum_{i=1}^{N} \left\{Z(d, g_{N}, Q_{N})(O_{i}) -  \theta_{N} \right\}^2
\end{align*}
is a consistent estimate of $\VAR(U^{d})$. For a recent review of influence functions, see \citep{hines2022demystifying}.

By applying cross-fitting (/sample splitting) of the nuisance models in combination with the doubly robust score,
the nuisance models can be estimated using flexible
machine learning methods without causing asymptotic bias
\citep{chernozhukov2018double}. This functionality is also
implemented in \code{policy_eval}. In Section \ref{sec:pol_eval_K} we present the estimating procedure
of the policy value in detail and generalize it to
multi-category policies over multiple stages.\\

In many situations the aim of the researcher is not just to evaluate a given policy,
but to learn the optimal policy from the data. In \pkg{polle}, this functionality is
implemented in the function \code{policy_learn()}.
The optimal treatment policy $d_0$ is defined as the policy for which it holds
theat $\E[U^d]\leq \E[U^{d_0}]$ for all other policies $d$. Thus, a direct
approach to policy learning is to use \eqref{eq:E_Z_d_K1} as a
loss function and perform value search:
\begin{align*}
d_N = \argmin_{d\in \mathcal{D}} \sum_{i=1}^{N}\tilde L(d)(g_N, Q_n)(O_i) = \argmin_{d \in \mathcal{D}} (-1) \sum_{i = 1}^N  Z(d, g_N, Q_n)(O_i).
\end{align*}
In practice, the complexity of the class of candidate policies $\mathcal{D}$ is
bounded for the value search to be viable. Examples include threshold policies
and policy trees, see \citep{athey2021policy}.
\code{policy_tree(type = 'ptl')} provides a wrapper for
policy tree learning in \pkg{polle}. In our experience, scalability and flexibility
of the current policy tree implementation can be an issue.
In the first part of Section \ref{sec:policylearn} we present the methodology in detail and
generalize it to multiple stages.\\

A key and rather intuitive result is that the optimal policy is also
identified as
\begin{align}
  d_0(h) &= \argmax_{a\in\{0,1\}} \E[U^a|H = h]\\
         &= \argmax_{a\in\{0,1\}} Q_0(h, a). \label{eq:concept_opt_pol}
\end{align}
This result motivates $Q$-learning which rely on estimating the $Q$-function.
The fitted $Q$-function is then plugged into \eqref{eq:concept_opt_pol} to get the associated estimated policy.
The problem with $Q$-learning is that we put too much faith in our ability to model the
$Q$-function. Due to confounding, in order to identify the causal optimal policy,
we have to input the complete history $H$ and estimate the $Q$-function consistently.
However, like value search, in the non-trivial case where $H$ is not discrete,
it is impossible to construct asymptotically reasonable estimators of the $Q$-function
without restricting the complexity of the class of candidate functions
\citep{semenova2021debiased,kennedy2020towards,luedtke2023one,nie2021quasi}.
In a sense, $Q$-learning combines two problems, a causal estimation problem and
a policy estimation problem,
which are hard to solve at the same time.

Luckily, it is possible to decompose the two
problems via the doubly robust score.
In the same process it is also possible to restrict the input to
the policy that we want to optimize. This is useful if we do not want to collect
the full history in a future implementation or if we want to exclude variables which
are unethical to collect or use as input to the policy.

For the purpose of restricting the input to the policy,
let $V\in \mathcal{V}$ be a subset or a function of the history
$H$. Let $d^V:\mathcal{V} \rightarrow \{0,1\}$ denote a $V$-restricted policy.
Finally, let $\mathcal{D}^V$ denote the class of $V$-restricted policies.
The optimal $V$-restricted policy is simply defined as the policy $d^V_0$ for
which it holds that $\E[U^{d^V}]\leq \E[U^{d^V_0}]$ for all
$d^V \in \mathcal{D}^V$. Similarly as above, the
optimal $V$-restricted policy is given by
\begin{align}
  d^V_0(v) &= \argmax_{a\in \{0,1\}} \E\left[U^a \big | V = v\right] \nonumber \\
  &= \argmax_{a\in \{0,1\}} \QV_0(v,a) \nonumber \\
  &= I\left\{ \QV_0(v, 1) - \QV_0(v, 0) > 0 \right\}, \label{eq:qvsinglestage}
\end{align}
see \citep{luedtke2016super}. The $\QV$-function is identified in two ways:
\begin{align*}
\QV_0(V,a) = \E\left[\frac{I\{A=a\}}{g_0(H, a)} U \Big | V \right] = \E\left[Q_0(H, a) | V \right].
\end{align*}
Again, the nuisance models can be combined to create a doubly robust expression
for the optimal $V$-restricted
policy. Define $Z(a, g, Q)$ as $Z(d, g, Q)$ from line \eqref{eq:Z_d_K1} under
the static policy $d(H) = a$.
If either $g = g_0$ or $Q = Q_0$ it holds that
\begin{align*}
\E\left[Z(a,g, Q)(O)\big | V \right] = \QV_0(V,a).
\end{align*}
This result directly inspires a doubly robust regression type estimator for the
$\QV$-function. Specifically, we let $\QV_N$ denote the function with the lowest
empirical mean squared error loss:
\begin{align*}
\QV_N(\cdot, a) &= \argmin_{QV}  \sum_{i=1}^{N} L(QV)(g_N, Q_N)(O_i) \\
&= \argmin_{QV} \sum_{i=1}^{N} \Big\{ Z\big(a, g_N, Q_N\big)\big(O_i\big) - QV\big(V_i, a\big)\Big \}^2.
\end{align*}
As we have detached estimation of the $Q$-function and the $\QV$-function,
we can truly regard the $Q$-function as a nuisance parameter and use flexible
machine learning methods to reduce bias.
At the same time we can also restrict the complexity of the learned $\QV$-function
without loosing causal interpretability. For example, we can let the $\QV$-function
be a member of a Donsker class such as the class of smooth
parametric models \citep{luedtke2020performance}.

In \pkg{polle}, doubly robust $\QV$/$Q$-learning is implemented in \code{policy_learn(type = 'drql')}. In Section \ref{sec:policylearn} we
generalize doubly robust $QV$/$Q$-learning to the multi-stage case.

For binary action sets, it is evident from \eqref{eq:qvsinglestage} that only
the contrast between the $QV$-functions is relevant for learning the optimal
policy. We denote this contrast as the blip, which in the single-stage case
is also known as the conditional average treatment effect (CATE):
\begin{align*}
B_{0}(v) = \QV_0(v, 1) - \QV_0(v, 0) = \E\left[U^{1}-U^{0}|V=v\right]
\end{align*}
A doubly robust loss function for the blip function $B_{0}$ is given by
\begin{align*}
L(B)(g_{0}, Q_{0})(O) = \Big\{ Z\big(1, g_{0}, Q_{0}\big)\big(O\big) - Z\big(0, g_{0}, Q_{0}\big)\big(O\big) - B\big(V\big)\Big\}^{2}.
\end{align*}
which is implemented in \code{policy_learn(type = 'blip')}.

An advantage of doubly robust $Q$-learning and blip learning is that we can take advantage of the extensive available
implementations for (regularized) regression. A downside of $\QV$-learning and blip learning is that these methods
do not target the decision boundary directly.
Instead, they minimize the estimation error of the given function.
This function approximation is then plugged into the threshold function to obtain an estimate of the decision boundary.
An alternative direct approach, like value search, is to use the weighted classification loss function given by
\begin{align*}
  L(d)(g_{0}, Q_{0})(O) = &\left| Z\big(1, g_{0}, Q_{0}\big)\big(O\big) - Z\big(0, g_{0}, Q_{0}\big)\big(O\big) \right| \\
 & \times I \left\{ d(V) \neq I \left\{ Z\big(1, g_{0}, Q_{0}\big)\big(O\big) - Z\big(0, g_{0}, Q_{0}\big)\big(O\big) > 0 \right\}  \right\},
\end{align*}
which we show to be equivalent to the value loss function in Appendix \ref{sec:class_loss}. Usually, the indicator function will be replaced by a convex
surrogate to ease minimization. The \pkg{polle} package wraps the classification functionality of the package \pkg{DTRlearn2} via \code{policy_learn(type='owl')},
though this is not based on the presented doubly robust score, but rather a different augmented score, see \citep{liu2018augmented}.

Since, the policy learners in \pkg{polle} are based on different loss functions,
we need an interpretable performance measure for a fair comparison between the policy learners.
As well as evaluating user-defined policies \code{policy_eval()} also allow for easy evaluation
of an arbitrary policy learner by returning the value of the estimated policy, see Section \ref{sec:policyperformance}
for more details. This is a key feature in \pkg{polle} which is not available in
other \proglang{R} packages.

The final concept that we want to introduce for now is
realistic policy learning. Positivity violations or even near positivity violations is a concern
for both policy learning and evaluation \citep{petersen2012diagnosing}.
If we in some stratum of the history do not observe
both treatments it is impossible to learn the optimal policy in the given
stratum without strong structural assumptions.
Thus we introduce the set of (estimated) realistic actions at level $\alpha$ as
\begin{align*}
D^\alpha_N(h) = \{a\in \mathcal{A}: g_N(h,a)>\alpha \}.
\end{align*}
$Q$-learning can then easily be adapted to the set of realistic actions as
follows
\begin{align*}
d_N(h) = \argmax_{a\in D^\alpha_N(h)} Q_N(h, a).
\end{align*}
For binary action sets, all of the presented policy learners can be adapted to
only recommend actions which are deemed realistic at a given level.

In the next section we formalize all of the above concepts and generalize them
to multiple stages.


\section{Setup and methods} \label{sec:methods}

\subsection{General multi-stage setup} \label{sec:setup}
Let $K \geq 1$ denote a fixed number of stages. Let $B \in \mathcal{B}$ denote the baseline covariates. For a finite set $\mathcal{A}$, let $A_{k} \in \mathcal{A}$ denote the decision or action at stage $k \in \{1,\ldots,K\}$. For $k \in \{1,\ldots,K + 1\}$, let $S_{k} \in \mathcal{S}$ denote the state at stage $k$. The trajectory for an observation can be written as
\begin{align*}
O = (B, S_{1}, A_{1}, S_{2}, A_{2},\ldots, S_{K}, A_{K}, S_{K+1}), 
\end{align*}
as illustrated in Figure \ref{fig:dtr2}. Usually, we will assume to have a sample of $N$ iid observations indexed as $\{O_i\}_{i\in 1...,N}$.
\begin{figure}[htbp]
\centerline{\includegraphics[width=0.65\textwidth]{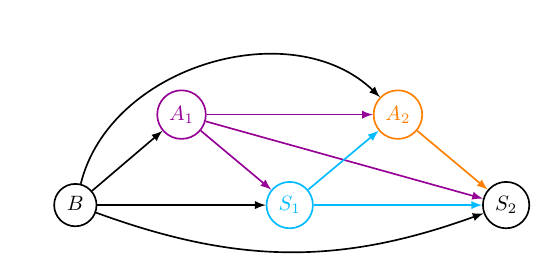}}
\caption{Graph for the observational data with two stages. \(B\) is a baseline
  covariate, \(A_{1}, A_{2}\) are the two decisions at stages 1 and 2, and
  \(S_{1}, S_{2}\) are the state variables. From each of the state variables, a
  reward can be derived, and the sum of these defines the
  utility of the decisions. \label{fig:dtr2}}
\end{figure}
For $k \in \{1,\ldots,K + 1\}$, let $\overline{S}_{k} = (S_{1},\ldots,S_{k})$,
$\overline{A}_{k} = (A_{1},\ldots, A_{k})$ and
$H_{k} = (B, \overline{S}_{k}, \overline{A}_{k-1}) \in \mathcal{H}_k$ define the
history where $A_{0} = A_{K+1} = \emptyset$.
Using the implied ordering, the density of the data can be expressed as
\begin{align}
p_0(O) = p_0(B)\left[ \prod_{k = 1}^{K} p_{0, k}(A_{k}|H_{k}) \right] \left[ \prod_{k = 1}^{K+1} p_{0, k}(S_{k}| H_{k-1}, A_{k-1}) \right]. \label{eq:likelihood_obs}
\end{align}

For convenience, let $S_k = (X_k, U_k)$, where $U_k \in \mathbb{R}$ is the $k$th reward, and $X_k$ is a state covariate/variable for $k \in \{1,\ldots,K\}$ and $X_{K+1} = \emptyset$. The utility is the sum of the rewards
\begin{align*}
U = \sum_{k = 1}^{K+1} U_{k}.
\end{align*}
In the \pkg{polle} package the function \code{policy_data()} helps the user specify the above data structure.

\subsection{Policy value estimation} \label{sec:pol_eval_K}
A policy is a set of rules $d = (d_1, ..., d_K)$, $d_k: \mathcal{H}_k \mapsto \mathcal{A}$ assigning an action in each stage.
Let $D_{0,k}(h_k)\subseteq \mathcal{A}$ denote the feasible set of decisions at stage $k$ for history $h_k$ under $P_0$, i.e.,
\begin{align*}
D_{0,k}(h_k) = \{a_k \in  \mathcal{A}: p_{0, k}(a_k|h_k) > 0\}.
\end{align*}
Define the class of feasible policies $\mathcal{D}_0$ as all sets of rules satisfying $d_k(h_k)\in D_{0,k}(h_k)$.

For a feasible policy $d$, let $P^d_0$ be the distribution with density
\begin{align}
p^{d}_0(O) = p_0(B)\left[ \prod_{k = 1}^{K} I\{A_k = d_k(H_k)\} \right] \left[ \prod_{k = 1}^{K+1} p_{0, k}(S_{k}| H_{k-1}, A_{k-1}) \right]. \label{eq:likelihood_d}
\end{align}

Let $O^{d}$ denote the data with distribution given by \eqref{eq:likelihood_d}, which is identified from the observed data. Define the value of the policy as
\begin{align*}
\theta^d_0 = \E[U^{d}].
\end{align*}
Under consistency and sequential randomization the above value will have a causal interpretation as the mean utility under an intervention given by the feasible policy.\\

The value of a feasible policy $d$ can explicitly be stated via the $Q$-functions recursively defined as
\begin{align*}
Q_{0,K}(h_K,a_K) = \E[U\mid H_{K} = h_{K}, A_K = a_K] 
\end{align*}
\begin{align*}
Q^{\underline{d}_{k+1}}_{0,k}(h_k, a_k) = \E\left[  Q^{\underline{d}_{k+2}}_{0,k+1}(H_{k+1}, d_{k+1}(H_{k+1})) \mid H_k = h_k, A_k = a_k  \right], \quad k \in \{1,\ldots,K-1\} 
\end{align*}
where $\underline{d}_k = (d_k, ..., d_K)$. It is possible to show that the
target parameter is identified as
\begin{align*}
\theta^d_0 = \E[Q^d_{0,1}(H_1, d_1(H_1)].
\end{align*}
The recursive structure of the $Q$-functions directly inspires a recursive regression procedure resulting in an estimate $Q^d_{N,1}$, based on $N$ iid observations. The value can then be estimated as the empirical mean of $Q^d_{N,1}(H_1, d_1(H_1)$.\\

The value of a feasible policy $d$ can also be stated via the $g$-functions defined as
\begin{align*}
g_{0,k}(h_k, a_k) = p_{0, k}(a_k\mid h_k),
\end{align*}
for $k \in \{1,\ldots,K\}$. Again, it is possible to show that
\begin{align*}
\theta^d_0 = \E\left[ \left( \prod_{k = 1}^{K} \frac{I\{A_{k} = d_k(H_{k})\}}{g_{0,k}(H_{k}, A_{k})} \right) U \right].
\end{align*}
Given regression estimates $g_{N,k}$, the value can now be estimated as the weighted empirical mean of the observed utilities.\\

Finally, it is possible to combine the two estimation approaches. Define the doubly robust policy score at stage $k$ as
\begin{align}
Z_k({\underline{d}_{k}},& g, Q^{\underline{d}_{k+1}})(O) \nonumber \\
&= \,Q_{k}^{\underline{d}_{k+1}}(H_k , d_k(H_k)) \nonumber \\
&+ \sum_{r = k}^K \left\{ \prod_{j = k}^{r} \frac{I\{A_j = d_j(H_{j})\}}{g_{j}(H_j, A_j)} \right\}\left\{Q_{r+1}^{\underline{d}_{r+2}}(H_{r+1} , d_{r+1}(H_{r+1})) - Q_{r}^{\underline{d}_{r+1}}(H_r , d_r(H_r))\right\}, \label{eq:Z_score_d_k}
\end{align}
where $Q_{K+1}(H_{K + 1} , d_{K+1}(H_{K+1})) = U$. It is possible to show that $\E[Z_1(d, g, Q^d)(O)] = \theta_0^d$ if either $g = g_0$ or $Q^d = Q^d_0$, see
for example \cite{tsiatis2019dynamic}. This result directly inspires a doubly robust moment type estimator of the policy value, see Algorithm \ref{alg:cross_fitted_value}.
In \pkg{polle}, this estimator is implemented in \code{policy_eval(type ='dr', policy)}.

\begin{algorithm}[]\label{alg:cross_fitted_value}
    \SetAlgoLined
    \DontPrintSemicolon
    \SetKwFunction{split}{M-folds}
    \SetKwInOut{Input}{input}
    \SetKwInOut{Output}{output}

    \Input{Data set with iid observations \(\mathcal{O} = (O_{1},\ldots,O_{N})\)\\ Feasible policy \(d\)\\
    Action probability regression procedure $\hat{g}$\\
    Outcome regression procedure $\hat{Q}^d$}
    \Output{Value estimate \(\theta^{d}_{N}\)\\
    Variance estimate $\Sigma^d_N$}
    \vspace*{1em}

    \(\{\mathcal{O}_1,\dots,\mathcal{O}_M\}\) = \split(\(\mathcal{O}\))\;

    \ForEach{\(m\in\{1,\ldots,M\}\)}{
    $g_m = \hat{g}(\mathcal{O} \setminus  \mathcal{O}_m)$\;
    $Q^d_m = \hat{Q}^d(\mathcal{O} \setminus  \mathcal{O}_m)$\;
    $\mathcal{Z}_{1,m} = \{Z_1(d, g_m, Q_m^d)(O): O\in \mathcal{O}_m \}$ \;
    }
    $\theta^d_N = N^{-1} \sum_{m = 1}^M \sum_{Z \in \mathcal{Z}_{1,m} } Z$ \;
    $\Sigma^d_N = N^{-1} \sum_{m = 1}^M \sum_{Z \in \mathcal{Z}_{1,m} } (Z - \theta^d_N)^2$\:

\caption{Cross-fitted doubly robust estimator of the policy value $\theta^d_0$}
\end{algorithm}

It is well known that $\psi_0^d(O) = Z_1(d, g_0, Q^d_0)(O) - \theta_0^d$ is the efficient influence function/curve for the policy value. We assume that the absolute utility is bounded and that $g_{k,0}(H_k, A_K) > \epsilon$ almost surely for some $\epsilon>0$. Let $\lVert g \rVert_{P,2} = \max_{j \in \{1, ...,K\}} \lVert g_j \rVert_{P,2}$ and $\lVert Q^d \rVert_{P,2} = \max_{j \in \{1, ...,K\}} \lVert Q^{\underline{d}_{j+1}}_j \rVert_{P,2}$. If, with probability converging to one, $g_{k,m}(H_k, A_k)> \epsilon$ and
\begin{align*}
\lVert g_m - g_0 \rVert_{P_0,2} =& o_{P_0}(1)\\
\lVert Q_m^d - Q^d_0 \rVert_{P_0,2} =& o_{P_0}(1)\\
\lVert g_m - g_0 \rVert_{P_0,2} \times \lVert Q_m^d - Q^d_0 \rVert_{P_0,2} =& o_{P_0}(N^{-1/2}),
\end{align*}
then
\begin{align*}
N^{1/2}(\theta^d_N - \theta^d_0) = N^{-1/2}\sum_{i = 1}^N \psi_0^d(O_i) + o_{P_0}(1).
\end{align*}
Thus, $\Sigma^d_N$ from Algorithm \ref{alg:cross_fitted_value} is a good estimate of the asymptotic variance of the value estimate if the nuisance models $\hat g$ and $\hat Q^d$ are correctly specified. It is important to note that the convergence rate conditions are relatively weak. For example, $Q^d_0$ and $g_0$ may be estimated at rate $o_{P_{0}}(N^{-1/4})$, which is much lower than a parametric rate of order $o_{P_{0}}(N^{-1/2})$. This result justifies the use of adaptive and regularized nuisance models, see \cite{chernozhukov2018double}.

\subsection{Policy learning}\label{sec:policylearn}

One of the main objective of \pkg{polle} is to learn the optimal policy from data.
As we want to be able to control the policy input,
we start by defining the optimal policy within a class of policies that are restricted to a subset of the observed history.
The following result is a generalization of \cite{van2014targeted}.\\

Let $V_k$ be a function (or subset) of $H_k$. A $V$-restricted policy is a set of rules $d^V = (d_{1}^V, ..., d_{K}^V)$, $d_{k}^V: \overline{\mathcal{A}}_{k-1}\times \mathcal{V}_k \mapsto \mathcal{A}$. Let $\mathcal{D}^V$ denote the class of $V$-restricted policies. Under positivity, i.e., $D_{0,k}(H_k) = \mathcal{A}$ almost surely, the $V$-optimal policy is defined as
\begin{align*}
d^V_0 = \arg \max_{d\in \mathcal{D}^V} \E[U^d].
\end{align*}

The following theorem specifies the $V$-optimal policy in a recursive manner. A proof for the two-stage case can be found in Appendix \ref{sec:v_optimal_policy}.
\begin{theorem} \label{theo:optimal_v_policy_K}
Under positivity, for any $a = (a_1, ..., a_K)$ and policy $d$ define
\begin{align}
\QV_{0,K}(\overline{a}_{K-1}, v_K, a_K) &= \E[U^{\overline{a}_{K}} | V_K^{\overline{a}_{K-1}} = v_K ] \label{eq:W_d_K},\\
\QV^{\underline{d}_{k+1}}_{0,k}(\overline{a}_{k-1}, v_k, a_k) &= \E[U^{\overline{a}_{k}, \underline{d}_{k+1}}| V_k^{\overline{a}_{k-1}} = v_k ] \quad k \in \{1,..., K-1\}. \label{eq:W_d_k}
\end{align}
If
\begin{align}
\E[U^{a} | V_1, ..., V_k^{\overline{a}_{k-1}}] = \E[U^{a} | V_k^{\overline{a}_{k-1}}], \quad k \in \{1,..., K\}, \label{eq:optimal_v_policy_k_condition}
\end{align}
then the $V$-optimal policy $d^V_{0}$ is recursively given by
\begin{align*}
d^V_{0,K}(\overline{a}_{K-1}, v_K) &= \arg\max_{a_K} \QV_{0,K}(\overline{a}_{K-1}, v_K, a_K),\\
d^V_{0,k}(\overline{a}_{k-1}, v_k) &= \arg\max_{a_k} \QV^{\underline{d}^V_{0, k+1}}_{0,k}(\overline{a}_{k-1}, v_k, a_k) \quad k \in \{1,..., K-1\}.
\end{align*}
\end{theorem}

If for all $k\in \{1,...,K\}$ and $r<k$, $V_r^{\overline{a}_{r-1}}$  is a function
of $V_k^{\overline{a}_{k-1}}$ then \eqref{eq:optimal_v_policy_k_condition} holds by construction. The intuition behind condition
\eqref{eq:optimal_v_policy_k_condition} is that it will ensure that the tower property holds for
the nested conditional expectations \eqref{eq:W_d_K} and  \eqref{eq:W_d_k}.
Finally, note that only future rewards affects the optimal decision at stage $k$ since
\begin{align*}
\arg\max_{a_k} \QV^{\underline{d}^V_{0, k+1}}_{0,k}&(\overline{a}_{k-1}, v_k, a_k)\\
=& \arg\max_{a_k} \E\left[U_1 + ... + U_{k}^{\overline{a}_{k-1}} + U_{k+1}^{\overline{a}_k} + ... + U_{K+1}^{\overline{a}_k, \underline{d}^V_{0,k+1}} \big | V_k^{\overline{a}_{k-1}} = v_k\right]\\
=& \arg\max_{a_k} \E\left[U_{k+1}^{\overline{a}_k} + ... + U_{K+1}^{\overline{a}_k, \underline{d}^V_{0,k+1}} \big | V_k^{\overline{a}_{k-1}} = v_k\right].
\end{align*}
The basis for learning the $V$-restricted optimal policy is to construct an observed data loss function
which identifies $d^V_0$, i.e., construct a function $L$ for which $\E[L(d)(O)]$ is minimized in $d^V_0$.
Various loss functions inspires different algorithms for estimating the $V$-optimal policy.
In the following, we present four different doubly robust loss functions, a value, quality, blip, and classification loss function.

\subsubsection{Value search}

For the final stage $K$ consider the loss function $\tilde L_{K}(d_K)(g_K, Q_K)(O)$ in $d_K$ given by
\begin{align}
-\tilde L_{K}(d_K)(g_K, Q_K)(O) &= Z_{K}(d_K, g, Q)(O) \nonumber \\
&= Q_{K}(H_k, d_K(H_k)) +  \frac{I\{A_K = d_K(H_k)\}}{g_{K}(H_K, A_K)} \left\{U - Q_{K}(H_K, A_K)\right\} \label{eq:valuelossK}
\end{align}
If either $Q_K = Q_{0,K}$ or $g_K = g_{0,K}$, then
\begin{align*}
\E\left[\tilde L_{K}(d_K)(g_K, Q_K)(O)\right] = -\E\left[U^{d_K}\right].
\end{align*}
Thus, for a $V$-restricted policy
\begin{align*}
\E\left[\tilde L_{K}(d^V_K)(g_K, Q_K)(O)\right] = -\E\left[\QV_{0,K}\left(\overline{A}_{K-1}, V_K, d^V_K(\overline{A}_{K-1}, V_K)\right)\right],
\end{align*}
meaning that over the class of $V$-restricted policies $\mathcal{D}^V_K$ the expected loss is minimized in $d^V_{0,K}$ by Theorem \ref{theo:optimal_v_policy_K}.\\

For stage $k \in \{1, ..., K-1\}$ consider the loss function $\tilde L_{k}(d_k)(\underline{d}_{k+1}, g, Q^{\underline{d}_{k+1}})(O)$ in $d_k$ given by
\begin{align*}
-\tilde L_{k}(d_k)(\underline{d}_{k+1}, g, Q^{\underline{d}_{k+1}})(O) = Z_{k}([d_k,\underline{d}_{k+1}], g, Q^{\underline{d}_{k+1}})(O).
\end{align*}
If either $Q^{\underline{d}_{k+1}} = Q^{\underline{d}_{k+1}}_{0}$ or $g = g_{0}$, then
\begin{align*}
E\left[\tilde L_{k}(d_k)(\underline{d}_{k+1}, g, Q^{\underline{d}_{k+1}})(O)\right] = -E\left[U^{d_k, \underline{d}_{k+1}}\right],
\end{align*}
and for a $V$-restricted policy at stage $k$ it holds that
\begin{align*}
E\left[\tilde L_{k}(d^V_k)(\underline{d}_{k+1}, g, Q^{\underline{d}_{k+1}})(O)\right] = -E\left[\QV^{\underline{d}_{k+1}}_{0,k}(\overline{A}_{k-1}, V_k, d^V_k(\overline{A}_{k-1}, V_k)\right]
\end{align*}
Thus, given $\underline{d}^V_{0,k+1}$, the above expected loss over $\mathcal{D}^V_k$ is again minimized in $d^V_{0,k}$ by Theorem \ref{theo:optimal_v_policy_K}.\\

The constructed loss function directly inspires recursive value search, see Algorithm \ref{alg:sequential_value_search}. Similar to Algorithm \ref{alg:cross_fitted_value}, the algorithm utilizes cross-fitted values of the nuisance models at each step. However, it is important to note that the $Q$-models are not truly cross-fitted (except for the last stage), because the fitted policy at a given stage depends on the fitted policy at later stages. A nested cross-fitting scheme would be required to make the folds (used to fit the $Q$-models) independent.

\begin{algorithm}[]\label{alg:sequential_value_search}
    \SetAlgoLined
    \DontPrintSemicolon
    \SetKwFunction{split}{L-folds}
    \SetKwInOut{Input}{input}
    \SetKwInOut{Output}{output}
    \newcommand{\forcond}{$k=K$ \KwTo $1$}

    \Input{Data set with iid observations \(\mathcal{O} = (O_{1},\ldots,O_{N})\)\\
          Class of $V$-restricted policies $\mathcal{D}^V$\\
          Function class minimization procedure $\hat F$\\
          Action probability regression procedure $\hat g$\\
          Outcome regression procedure $\hat Q = \{\hat Q_1, ..., \hat Q_K\}$}
    \Output{$V$-restricted optimal policy estimate $d^V_N$}
    \vspace*{1em}

    \(\{\mathcal{O}_1,\dots,\mathcal{O}_L\}\) = \split(\(\mathcal{O}\))\;
     \ForEach{\(l\in\{1,\ldots,L\}\)}{
        $g_l = \hat{g}(\mathcal{O} \setminus  \mathcal{O}_l)$\;
     }
    \For{\forcond}{
        \ForEach{\(l\in\{1,\ldots,L\}\)}{
          $Q^{\underline{d}^V_{N, k+1}}_{l, k} = \hat{Q}_k^{\underline{d}^V_{N, k+1}}\left(\left\{H_k, A_k, Q^{\underline{d}^V_{N, k+2}}_{k+1, l}\left(H_{k+1}, d^V_{N, k+1}(H_{k+1})\right): O \in \mathcal{O} \setminus  \mathcal{O}_l \right\}\right)$\;
      }
      $d^V_{N, k} =  \hat F_{d^V_k \in \mathcal{D}^V}\left( \sum_{l=1}^L \sum_{O \in \mathcal{O}_l} \tilde{L}(d^V_k)\big(\underline{d}^V_{k+1, N}, g_l, Q^{\underline{d}^V_{k+1,N}}_{l}\big)(O) \right)$ \;
    }
\caption{Recursive Value Search}
\end{algorithm}

Algorithm \ref{alg:sequential_value_search} requires a suitable function class minimization procedure $\hat F$ imitating
$\arg\min_{d^V_k \in \mathcal{D}^V}\{\cdot\}$ at every stage. The \proglang{R} package \pkg{policytree},
see \cite{sverdrup2020policytree}, implements such a minimization procedure where the class of policies is given by
decision trees. See \cite{zhou2018offline} for theoretical results related to this implementation.
In \pkg{polle}, recursive value search using \pkg{policytree} is implemented in \code{policy_learn(type = 'ptl')}.

\subsubsection{Quality learning}

Under positivity, for any $a = (a_1, ..., a_K)$ and any policy $d$, let $Z_{k}([a_k,\underline{d}_{k+1}], g, Q^{\underline{d}_{k+1}})(O)$ be given by \eqref{eq:Z_score_d_k} with $d_k$ replaced by the static policy $a_k \in \mathcal{A}$. For the final stage $K$, define $\QV_{0,K}(a_K)(\overline{a}_{K-1}, v_K) = \QV_{0,K}(\overline{a}_{K-1}, v_K, a_K)$  from equation \eqref{eq:W_d_K}. If $g_K = g_{0,K}$ or $Q_K = Q_{0,K}$ then
\begin{align*}
\E\left[Z_{K}(a_K, g, Q)(O)\big| \overline{A}_{K-1}, V_K \right] = \QV_{0,K}(a_K)(\overline{A}_{K-1}, V_K).
\end{align*}
Now, a valid loss function for $\QV_{0,K}(a_K)$ over functions $\QV_K: \overline{\mathcal{A}}_{K-1}\times \mathcal{V}_K \mapsto \mathcal{A}$ is given by
\begin{align*}
L_{K}(\QV_K)(a_K, g, Q)(O) = \left \{Z_{K}(a_K, g, Q)(O) - \QV_{K}(\overline{A}_{K-1}, V_K)\right \}^2.
\end{align*}
Hence, any regression type estimator which minimizes the (empirical) mean squared error can be used to estimate $\QV_{0,K}(a_K)$. This can be repeated for every $a_K \in \mathcal{A}$. By Theorem \ref{theo:optimal_v_policy_K}, the $V$-optimal policy $d^V_{0, K}$ is then identified as $\arg\max_{a_K\in \mathcal{A} }\QV_{0,K}(a_K)$.\\

For $k \in \{1,..., K-1\}$, let $\QV^{\underline{d}_{k+1}}_{0,k}(a_K)(\overline{a}_{k-1}, v_k) = \QV^{\underline{d}_{k+1}}_{0,k}(\overline{a}_{k-1}, v_k, a_k)$ from equation \eqref{eq:W_d_k}. If $g = g_0$ or $Q^{\underline{d}_{k+1}} = Q^{\underline{d}_{k+1}}_0$ then
\begin{align*}
\E\left[Z_{k}([a_k, \underline{d}_{k+1}], g, Q^{\underline{d}_{k+1}})(O) \big| \overline{A}_{k-1}, V_k\right] = \QV^{\underline{d}_{k+1}}_{0,k}(a_k)(\overline{A}_{k-1}, V_k),
\end{align*}
and a valid loss function for $\QV^{\underline{d}_{k+1}}_{0,k}(a_k)$ over functions $\QV_k$ is given by
\begin{align*}
L_{k}(\QV_k)(a_k, \underline{d}_{k+1}, g, Q^{\underline{d}_{k+1}})(O) = \left \{ Z_{k}([a_k, \underline{d}_{k+1}], g, Q^{\underline{d}_{k+1}})(O) - \QV_{k}(\overline{A}_{k-1}, V_k)\right\}^2.
\end{align*}
Thus, given the future $V$-restricted optimal policy rules $\underline{d}^V_{0, k+1}$, if $g = g_0$ or $Q^{\underline{d}^V_{0, k+1}} = Q^{\underline{d}^V_{0, k+1}}_0$, then the expected loss is minimized in $\QV^{\underline{d}^V_{0, k+1}}_{0,k}(a_k)$. Again, this can be repeated for each $a_k \in \mathcal{A}$ and the $V$-optimal policy at stage $k$ is identified as $\arg\max_{a_k\in \mathcal{A}} \QV^{\underline{d}^V_{0, k+1}}_{0,k}(a_k)$.\\
The constructed quality loss function directly inspires doubly robust $V$-restricted $Q$-learning,
see Algorithm \ref{alg:dr_QV_learning}. In \pkg{polle} this policy estimator is implemented in \code{policy_learn(type = 'drql')}.

\begin{algorithm}[]\label{alg:dr_QV_learning}
    \SetAlgoLined
    \DontPrintSemicolon
    \SetKwFunction{split}{L-folds}
    \SetKwInOut{Input}{input}
    \SetKwInOut{Output}{output}
    \newcommand{\forcond}{$k=K$ \KwTo $1$}

    \Input{Data set with iid observations \(\mathcal{O} = (O_{1},\ldots,O_{N})\)\\
          Action probability regression procedure $\hat g$\\
          Outcome regression procedure $\hat Q = \{\hat Q_1, ..., \hat Q_K\}$\\
          Outcome regression procedure $\widehat{\QV} = \{\widehat{\QV}_1, ..., \widehat{\QV}_K\}$}
    \Output{$V$-restricted optimal policy estimate $d^V_N$}
    \vspace*{1em}

    \(\{\mathcal{O}_1,\dots,\mathcal{O}_L\}\) = \split(\(\mathcal{O}\))\;
     \ForEach{\(l\in\{1,\ldots,L\}\)}{
        $g_l = \hat{g}(\mathcal{O} \setminus  \mathcal{O}_l)$\;
     }
    \For{\forcond}{
        \ForEach{\(l\in\{1,\ldots,L\}\)}{
          $Q^{\underline{d}^V_{N,k+1}}_{k, l} = \hat{Q}_k\left(\left\{H_k, A_k, Q^{\underline{d}^V_{N,k+2}}_{k+1,l}\left(H_{k+1}, d^V_{N, k+1}(H_{k+1})\right): O \in \mathcal{O} \setminus  \mathcal{O}_l \right\}\right)$\;
          \ForEach{$a_k \in \mathcal{A}$}{
          $\widetilde{\mathcal{O}}_{k,l}(a_k) = \left\{\overline{A}_{k-1}, V_k, Z_k\left([a_k, \underline{d}^V_{N,k+1}], g_l, Q_l^{\underline{d}^V_{N, k+1}}\right)(O): O\in \mathcal{O}_l \right\}$ \;}
    }
    \ForEach{$a_k \in \mathcal{A}$}{
    $\QV^{\underline{d}^V_{N, k+1}}_{N, k}(a_k) = \widehat{\QV}_k\left(\{\widetilde{O} \in \widetilde{\mathcal{O}}_{k,l}(a_k): l = 1,...,L\}\right)$ \;
    }
    $d^V_{N,k} = \arg \max_{a_k \in \mathcal{A}} \QV^{\underline{d}^V_{N, k+1}}_{N, k}(a_k)$\;
    }
\caption{Doubly Robust $Q$-learning}
\end{algorithm}

\subsubsection{Blip learning}

As outlined in Section \ref{sec:concepts}, for a binary action set $\mathcal{A} = \{0,1\}$, instead of learning the $\QV$-function for each action,
is it sufficient to learn the contrast or blip between the functions. Considering Theorem \ref{theo:optimal_v_policy_K}, define
the blips as
\begin{align*}
B_{0,K}(\overline{a}_{K-1}, v_K) = \QV_{0,K}(\overline{a}_{K-1}, v_K, 1) -  \QV_{0,K}(\overline{a}_{K-1}, v_K, 0),\\
B_{0, k}^{\underline{d}_{k+1}}(\overline{a}_{k-1}, v_k) = \QV^{\underline{d}_{k+1}}_{0,k}(\overline{a}_{k-1}, v_k, 1) - \QV^{\underline{d}_{k+1}}_{0,k}(\overline{a}_{k-1}, v_k, 0).
\end{align*}
Using the blips, the $V$-restricted optimal policy is identified as
\begin{align*}
d^V_{0,K}(\overline{a}_{K-1}, v_K) &= I \left\{ B_{0,K}(\overline{a}_{K-1}, v_K) > 0 \right\},\\
d^V_{0,k}(\overline{a}_{k-1}, v_k) &= I \left\{ B^{\underline{d}^V_{0, k+1}}_{0,k}(\overline{a}_{k-1}, v_k) > 0 \right\} \quad k \in \{1,..., K-1\}.
\end{align*}
Doubly robust blip learning is now almost identical to Algorithm \ref{alg:dr_QV_learning}, with $\widehat QV$ replaced by $\hat B$ and
the doubly robust scores $Z_k\left([a_k, \underline{d}^V_{k+1}], g, Q^{\underline{d}^V_{k+1}} \right)(O)$ replaced by the contrasts
$$
W_{k}(\underline{d}^V_{k+1}, g, Q^{\underline{d}^V_{k+1}})(O) = Z_k([1, \underline{d}^V_{k+1}], g, Q^{\underline{d}^V_{k+1}})(O)- Z_k([0, \underline{d}^V_{k+1}], g, Q^{\underline{d}^V_{k+1}} )(O).
$$
In \pkg{polle}, blip learning is implemented in \code{policy_learn(type = 'blip')}.

\subsubsection{Weighted classification}

In this section we again consider the case where the action set is binary, i.e.,
$\mathcal{A} = \{0,1\}$. As shown in Appendix \ref{sec:class_loss}, a valid weighted classification
(0-1) loss function is given by
\begin{align}
L_{k}(d_k)(\underline{d}_{k+1}, g, Q^{\underline{d}_{k+1}})(O) = &\left \lvert W_{k}(\underline{d}_{k+1}, g, Q^{\underline{d}_{k+1}})(O) \right \rvert \nonumber \\
& \times I\Big\{d_k(\overline{A}_{K-1}, V_k) \neq  I \Big \{W_{k}(\underline{d}_{k+1}, g, Q^{\underline{d}_{k+1}})(O) > 0 \Big\}\Big \}. \label{eq:weighted_loss_function_k}
\end{align}
Given the future $V$-optimal policy rules $\underline{d}^V_{0, k+1}$, if $g = g_0$ or $Q^{\underline{d}^V_{0, k+1}} = Q^{\underline{d}^V_{0, k+1}}_0$,
the expected weighted classification loss function is minimized in $d^V_{0, k}$ over $\mathcal{D}^V_k$.

It can be challenging to perform minimization of the weighted classification loss function due to the
discontinuity of the indicator function. For this reason, it is common to use a convex surrogate
of the indicator function. Let $f_k: \overline{\mathcal{A}}_{K-1}, \times \mathcal{V}_k \mapsto \mathbb{R}$ be some action function
corresponding to $d_k$, i.e., $d_k(\overline{A}_{K-1}, V_k) = I\{f_k(\overline{A}_{K-1},V_k) >0\}$. Then
\begin{align*}
{L}_{k}(f_k)(\underline{d}_{k+1}, g, Q^{\underline{d}_{k+1}})(O) =& \left \lvert {W}_{k}(\underline{d}_{k+1}, g, Q^{\underline{d}_{k+1}})(O)\right \rvert \\
& \times I \left\{ f_k(\overline{A}_{K-1},V_k)\left[2 I \left \{{W}_{k}(\underline{d}_{k+1}, g, Q^{\underline{d}_{k+1}})(O)> 0\right\} -1 \right] \leq 0 \right\}
\end{align*}
is equivalent to \eqref{eq:weighted_loss_function_k}. Replacing $I\{x\leq 0\}$ in the above expression with a convex surrogate $\phi: \mathbb{R}\mapsto [0, \infty)$ differentiable in $0$ with $\phi'(0)<0$ yields a convex loss function given by
\begin{align*}
{L}_{k}^{\phi}(f_k)(\underline{d}_{k+1}, g, Q^d)(O) =& \left \lvert {W}_{k}(\underline{d}_{k+1}, g, Q^{\underline{d}_{k+1}})(O)\right \rvert\\
& \times \phi \left( f_k(\overline{A}_{K-1},V_k)\left[2 I \left \{W_{k}(\underline{d}_{k+1}, g, Q^{\underline{d}_{k+1}})(O)> 0\right\} -1 \right]\right).
\end{align*}
If $g = g_0$ or $Q^{d^V_0} = Q^{d^V_0}_0$ and the non-exceptional law holds, i.e., that
\begin{align}
0 < \E \left[W_{k}(\underline{d}_{k+1}, g, Q^{\underline{d}_{k+1}})(O)\right], \label{eq:non-excep}
\end{align}
then the expected weighted surrogate loss function is minimized in $f^V_{0,k}$ over the class of $V$-restricted action functions and $d^V_{0,k} = I\{f^V_{0,k} > 0\}$. The above result directly inspires recursive learning of
the restricted optimal policy using weighted classification methods similar to
Algorithm \ref{alg:sequential_value_search}.

The classification perspective was first established by \cite{zhao2012estimating}
and \cite{zhang2012estimating}. Various methods within this approach has been
implemented in the \proglang{R} packages \pkg{DTRlearn2} and \pkg{DynTxRegime}, see \cite{cran:DTRlearn2}
and \cite{cran:dyntxregime}. Generalizations to multiple actions (more than two)
has also been developed, see \cite{zhang2020multicategory}.

\pkg{polle} wraps the weighted classification function \code{owl()} from the \pkg{DTRlearn2} package, though
this algorithm is not based on the presented doubly robust blip score, but a different augmented score, see \citep{liu2018augmented}.
The policy learner is available via \code{policy_learn(type = 'owl')}.

\subsection{Learned policy performance} \label{sec:policyperformance}

Another main objective of \pkg{polle} is to be able to compare the performance of various policy learners.
For this purpose we advocate for targeting the estimated policy's value \citep{chakraborty2014inference}.
Conditional on a policy estimate $d_{N}$ the target parameter is given by
\begin{align*}
\theta_{0}^{d_{N}}= \E[U^{d_{N}}].
\end{align*}
This data-adaptive parameter arguably more practically relevant than the true optimal policy value
because the true optimal policy will never actually be implemented. Moreover the data-adaptive policy
value does not suffer from non-regularity issues related to the optimal policy value, see
\citep{hirano2012impossibility, luedtke2016statistical, robins2014discussion, chakraborty2013statistical}.

Under regularity conditions similar to those stated in Section \ref{sec:pol_eval_K}, see \citep{van2014targeted},
if $d_{N}$ has a limiting policy $d'$ then the efficient influence function for $\theta_{0}^{d_{N}}$ is given
by
\begin{align*}
Z_{1}(d', g_{0}, Q_{0}^{d'})(0)-\theta_{0}^{d'}.
\end{align*}
As described in Section \ref{sec:pol_eval_K} this allow us to construct a doubly robust estimator for the
value of the learned policy, see Algorithm \ref{alg:cross_fitted_learned_value}. The algorithm also provide a variance estimate enabling a
Wald type confidence interval. In the \pkg{polle} package, learned policy evaluation is implemented in the function \code{policy_eval(policy_learn)}.

\begin{algorithm}[]\label{alg:cross_fitted_learned_value}
    \SetAlgoLined
    \DontPrintSemicolon
    \SetKwFunction{split}{M-folds}
    \SetKwInOut{Input}{input}
    \SetKwInOut{Output}{output}

    \Input{Data set with iid observations \(\mathcal{O} = (O_{1},\ldots,O_{N})\)\\
    Policy learning procedure \(\hat{d}\)\\
    Action probability regression procedure $\hat{g}$\\
    Outcome regression procedure $\hat{Q}^d$}
    \Output{Cross-fitted value estimate \(\theta^{d}_{N}\)\\
    Cross-fitted variance estimate $\Sigma^d_N$}
    \vspace*{1em}

    \(\{\mathcal{O}_1,\dots,\mathcal{O}_M\}\) = \split(\(\mathcal{O}\))\;

    \ForEach{\(m\in\{1,\ldots,M\}\)}{
    $d_m = \hat{d}(\mathcal{O} \setminus  \mathcal{O}_m)$\;
    $g_m = \hat{g}(\mathcal{O} \setminus  \mathcal{O}_m)$\;
    $Q^{d_m}_m = \hat{Q}^{d_m}(\mathcal{O} \setminus  \mathcal{O}_m)$\;
    $\mathcal{Z}_{1,m} = \{Z_1(d_m, g_m, Q_m^{d_m})(O): O\in \mathcal{O}_m \}$ \;
    }
    $\theta^{d_N}_N = N^{-1} \sum_{m = 1}^M \sum_{Z \in \mathcal{Z}_{1,m} } Z$ \;
    $\Sigma^{d_N}_N = N^{-1} \sum_{m = 1}^M \sum_{Z \in \mathcal{Z}_{1,m} } (Z - \theta^{d_N}_N)^2$\:

\caption{Cross-fitted doubly robust estimator of $\theta^{d_N}_0$}
\end{algorithm}

\subsection{Stochastic number of stages}

The methodology developed for a fixed number of stages can be extended to handle a stochastic number of stages assuming that the maximal number of stages is finite. The key is to modify each observation such that every observation has the same number of stages.\\

Let $K^*$ denote the stochastic number of stages bounded by a maximal number of stages $K$. As in Section \ref{sec:setup}, let $B^*$ denote the baseline data, $(A_1^*, \ldots, A^*_{K^*})$ denote the decisions and $(S_1^*, \ldots, S^*_{K^*+1})$ denote the stage summaries where $S^*_k = (X_k^*, U^*_k)$ and $X^*_{K^*+1} = \emptyset$. The utility $U^*$ is still the sum of the rewards $U^* = \sum_{k=1}^{K^*+1} U^*_k$. We assume that the distribution of the observed data is given by $P^*_0$ composed of conditional densities $p_0^*(B^*)$, $p^*_{0,k}(A^*_k\mid H^*_k)$ for $k\in \{1,...,K\}$ and $p^*_{0,k}(S^*_k\mid H^*_k)$ for $k\in \{1,...,K+1\}$ such that the likelihood for an observation $O^*$ is given by
\begin{align*}
p^*_0(O^*) = p^*_0(B^*)\left[ \prod_{k = 1}^{K^*} p^*_{0, k}(A^*_{k}|H^*_{k}) \right] \left[ \prod_{k = 1}^{K^*+1} p^*_{0, k}(S^*_{k}| H^*_{k-1}, A^*_{k-1}) \right]. 
\end{align*}
For a feasible policy $d^* = (d^*_1, ..., d^*_K)$ the distribution $P^{^*d^*}$ is defined similar to $P^d$ in \eqref{eq:likelihood_d}. Also, the value under $P^{^*d^*}$ is defined as $\E\left[U^{*d} \right]$.\\

We now construct auxiliary data such that each observation $O^*$ has $K$ stages. Let $A_k = A^*_{k}$ for $k \leq K^*$ and $A_{k} = a^{\dagger} \in \mathcal{A}$ (for some default value $a^{\dagger}$) for $k > K^*$. Similarly, let $S_k = S^*_{k}$ for $k \leq K^* + 1$. Finally, let $X_{k} = \emptyset$ and $U_{k} = 0$ for $k > K^* + 1$ such that $U = U^*$. This construction implies a partly degenerate distribution $P_0$ over the maximal number of stages with density on the form given by \eqref{eq:likelihood_obs}, see \cite{goldberg2012q}. A feasible policy $d$ associated with $d^*$ is given by
\begin{align*}
d_k(H_k) =
\begin{cases}
a^{\dagger} &  \text{if}\, X_k = \emptyset \\
d^*_k(H_k) & \text{otherwise}.
\end{cases}
\end{align*}
Furthermore, it holds by construction that $g_{0,k}(H_k, a^{\dagger}) = 1$ and $Q^d_{0,k}(H_k, a^{\dagger}) = U$ if $X_k = \emptyset$.
Finally, a generalization of the results in \cite{goldberg2012q} yields that $\E\left[U^{d} \right] = \E\left[U^{*d^*} \right]$.
Thus, the methodology developed for a fixed number of stages can be used on the augmented data.

\subsection{Partial policy}

It may occur that a small subset of the observations has numerous stages. Without further structural assumptions, information about these late stages will be sparse. Uncertain estimation of the $Q$-functions for the late stages can be avoided by considering partial policies. Let $\tilde K < K$ and let $\overline{d}_{\tilde K}$ be a given policy up till stage $\tilde K$. A partial (stochastic) policy is now given by $(\overline{d}_{\tilde K}, A_{\tilde K + 1}, \ldots, A_{K})$. By setting $Q_{0,\tilde K} = \E[U\mid H_{\tilde K} = h_{\tilde K}, A_{\tilde K} = a_{\tilde K}]$ and $Q_{0,\tilde K+1} = U$, the efficient influence score for the partial policy value will be equal to \eqref{eq:Z_score_d_k} with $K$ replaced by $\tilde K$. From a practical point of view, implementation of a partial policy requires that $(g_{0,\tilde K + 1}, \ldots, g_{0,K})$ is known (or at least well approximated). It is easy to consider partial policies in the \pkg{polle} package by using \code{partial()} on a given policy data object.

\subsection{Realistic policy learning}

Positivity violations or even near positivity violations is a large concern for policy learning based on
historical data. Estimation of a valid loss function will solely rely on extrapolation of the $Q$-functions
to decisions with little or no support in the observed data, see \cite{petersen2012diagnosing}.
To address this issue we suggest restricting the set of possible interventions based on the action
probability model. For a probability threshold $\alpha>0$, define the set of realistic actions at stage $k$
based on the action probability model $g$ as
\begin{align*}
D^{\alpha}_{g,k}(h_k) = \{a_k \in  \mathcal{A}: g_{k}(h_k, a_k) > \alpha\}.
\end{align*}
It is relatively simple to modify doubly robust $Q$-learning to only consider realistic policies, see Algorithm \ref{alg:dr_RQV_learning}.
On the other hand, it is harder to make the same practical modification to a
given value search algorithm because the structure of the candidate function class
$\mathcal{D}^V$ changes in a non-trivial way.
However, in the situation that
the action set is dichotomous, the recommended action can be overruled by the alternative action,
if it is deemed unrealistic. In the \pkg{polle} package realistic policy learning is implemented via \code{policy_learn(alpha)}.

\begin{algorithm}[] \label{alg:dr_RQV_learning}
    \SetAlgoLined
    \DontPrintSemicolon
    \SetKwFunction{split}{L-folds}
    \SetKwInOut{Input}{input}
    \SetKwInOut{Output}{output}
    \newcommand{\forcond}{$k=K$ \KwTo $1$}

    \Input{Data set with iid observations \(\mathcal{O} = (O_{1},\ldots,O_{N})\)\\
          Action probability regression procedure $\hat g$\\
          Outcome regression procedure $\hat Q = \{\hat Q_1, ..., \hat Q_K\}$\\
          Outcome regression procedure $\hat{\QV} = \{\hat{\QV}_1, ..., \hat{\QV}_K\}$}
    \Output{Realistic $V$-restricted optimal policy estimate $d^V_N$}
    \vspace*{1em}
    $g_{N} = \hat{g}(\mathcal{O})$\;
    \(\{\mathcal{O}_1,\dots,\mathcal{O}_L\}\) = \split(\(\mathcal{O}\))\;
     \ForEach{\(l\in\{1,\ldots,L\}\)}{
        $g_l = \hat{g}(\mathcal{O} \setminus  \mathcal{O}_l)$\;
     }
    \For{\forcond}{
        \ForEach{\(l\in\{1,\ldots,L\}\)}{
          $Q^{\underline{d}^V_{N,k+1, l}}_{k, l} = \hat{Q}_k\left(\left\{H_k, A_k, Q^{\underline{d}^V_{N, k+2, l}}_{k+1,l}\left(H_{k+1}, d^V_{N,k+1,l}(H_{k+1})\right): O \in \mathcal{O} \setminus  \mathcal{O}_l \right\}\right)$\;
          \ForEach{$a_k \in \mathcal{A}$}{
          $\widetilde{\mathcal{O}}_{k,l}(a_k) = \left\{\overline{A}_{k-1}, V_k, Z_k\left([a_k, \underline{d}^V_{N,k+1,l}], g_l, Q_l^{\underline{d}^V_{N,k+1,l}}\right)(O): O\in \mathcal{O}_l \right\}$ \;}
    }
    \ForEach{$a_k \in \mathcal{A}$}{
    $\QV^{\underline{d}^V_{N, k+1}}_{N, k}(a_k) = \hat{\QV}_k\left(\{\widetilde{O} \in \widetilde{\mathcal{O}}_{k,l}(a_k): l = 1,...,L\}\right)$ \;
    }
    \ForEach{\(l\in\{1,\ldots,L\}\)}{
      $d^V_{N,k, l} = \underset{\quad a_k \in D^{\alpha}_{g_{l},k}}{\argmax} \QV^{\underline{d}^V_{N, k+1}}_{N, k}(a_k)$\;
    }
    $d^V_{N,k} = \underset{\quad a_k \in D^{\alpha}_{g_{N},k}}{\argmax} \QV^{\underline{d}^V_{N, k+1}}_{N, k}(a_k)$\;
    }
\caption{Realistic Doubly Robust $Q$-learning}
\end{algorithm}


\FloatBarrier
\section{Syntax and implementation details}  \label{sec:syntax}

The \pkg{polle} implementation is build up around four functions:
\code{policy\_data()}, \code{policy\_def()}, \code{policy\_eval()} and
\code{policy\_learn()}. Figure \ref{fig:flow} provides an overview of how the
functions relate and the main required inputs and outputs.
\begin{figure}[!htbp]
  \centering
  \includegraphics[width=1\textwidth]{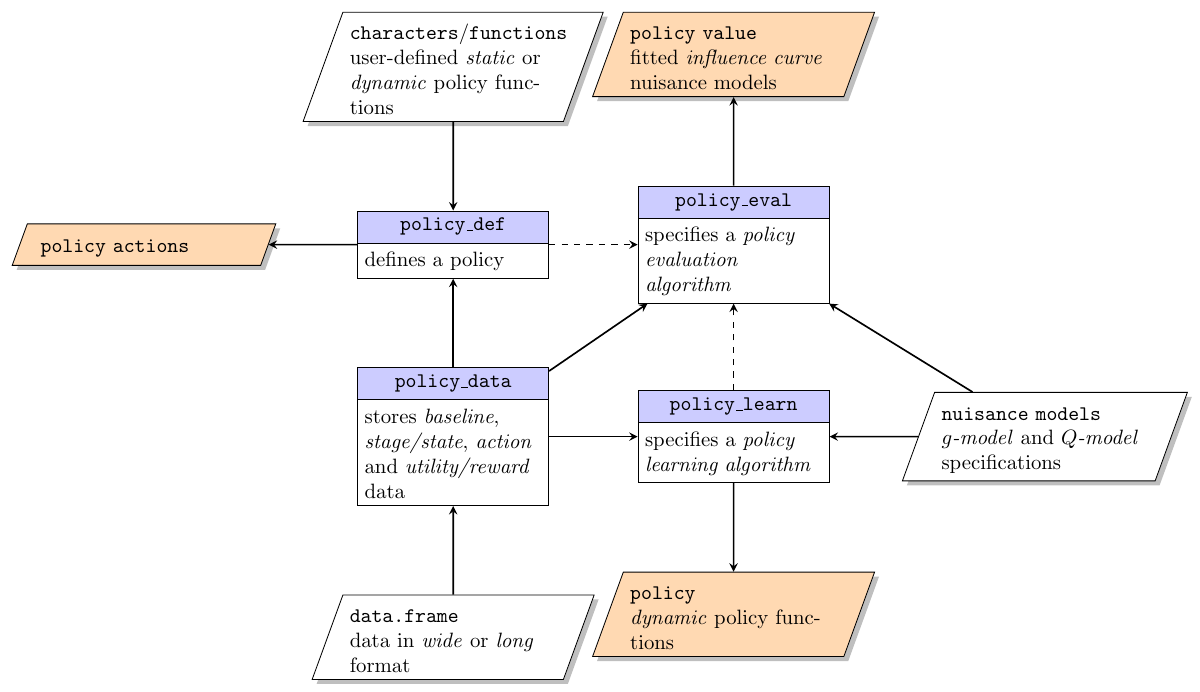}
  \caption{Overview of the four main functions of \pkg{polle} and their arguments
    and return values. The starting
    point is to define the input data in the correct format using
    \code{policy\_data()}. A policy can subsequently be defined directly by the
    user, \code{policy\_def}(), or estimated with one of the algorithms descripted
    in Section \ref{sec:policylearn} with \code{policy\_learn()}. The value of a
    policy can be estimated directly using \code{policy\_eval()}.
    \label{fig:flow}}
\end{figure}
\code{policy_data()} constructs a policy data object. The data input can be on
long or wide format. Usually, the wide format is used for
applications with a fixed number of stages and a possibly varying set of state
covariates. Assume that the observed data has the sequential form
\begin{align*}
O = (B, X_1, U_1, A_{1}, X_{2}, W_{2}, U_2, A_{2}, U_3),
\end{align*}
where $B$ is a baseline covariate, $X_1$, $X_2$, and $W_2$ are state covariates,
$U_1$, $U_2$, and $U_3$ are rewards, and $A_1$ and $A_2$ are actions.
Given a \code{data.table}/\code{data.frame} denoted \code{data} with
variable/column names \code{B}, \code{X_1}, \code{U_1}, \code{A_1}, \code{X_2},
\code{W_2}, \code{U_2}, \code{A_2} and \code{U_3}, we can apply
\code{policy_data} in the following way:
\begin{CodeChunk}
\begin{CodeOutput}
  policy_data(data,
              action = c("A_1", "A_2"),
              baseline = c("B"),
              covariates = list(X = c("X_1", "X_2"),
                                W = c(NA, "W_2")),
              utility = c("U_1", "U_2", "U_3"),
              type = "wide")
\end{CodeOutput}
\end{CodeChunk}
If only the final utility \code{U} is provided, we may replace
\code{c("U_1", "U_2", "U_3")} with \code{c("U")}. Note that each row in \code{data}
corresponds to a single observation.

The long format is inspired by the data format used for survival
data \citep{survival-package}. This format is relevant for handling a high
and possibly stochastic number of stages. Assume that the observed data has
the sequential form
\begin{align*}
O = (B, X_1, U_1, A_{1}, \ldots , X_{K^*}, U_{K^*}, A_{K^*}, U_{(K^*+1)}),
\end{align*}
where $K^*$ is the (stochastic) number of stages. Assume that \code{stage_data}
is a \code{data.table} with variable names \code{id}, \code{stage},
\code{event}, \code{X}, \code{U} and \code{A}. \code{id} and \code{stage}
denotes the observation ID and stage number. The variable \code{event} is an
event indicator which is \code{0} in stage $1$ through $K^*$ and \code{1}
in stage $(K^*+1)$. Also assume that \code{baseline_data} is a \code{data.table}
with variable names \code{id} and \code{B}. An application of \code{policy_data()}
is now given by:
\begin{CodeChunk}
\begin{CodeOutput}
  policy_data(stage_data,
              baseline_data = baseline_data,
              action = "A",
              baseline = c("B"),
              covariates = c("X"),
              utility = "U",
              id = "id",
              stage = "stage",
              event = "event",
              type = "long")
\end{CodeOutput}
\end{CodeChunk}
Note, an observation with $K^*$ stages spans over $(K^*+1)$ rows in
\code{stage_data} and a single row in \code{baseline_data}.

The function \code{policy_def()} constructs a user-specified static or
dynamic policy. The resulting policy object can be applied directly on a policy
data object or as input to \code{policy_eval()}.
\begin{CodeChunk}               %
\begin{CodeOutput}
  policy_def(policy_functions,
             full_history = FALSE,
             replicate = FALSE)
\end{CodeOutput}
\end{CodeChunk}
\code{policy_functions} may be a single function/character string or a
list of functions/character strings defining the policy at each stage.
The argument \code{full_history} defines the input to the policy functions.
If \code{full_history = FALSE} the state/Markov type history $(B, X_k)$ is
passed on to the functions with variable names \code{B} and \code{X}. If
\code{full_history = TRUE},
the full history $(B, X_1, A_1,\ldots, X_{k-1}, A_{k-1}, X_k)$ with variable
names \code{B, X_1, A_1,..., X_(k-1), A_(k-1), X_k} is passed on to the functions.
As an example, \code{function(X) 1*(X>0)} in combination with
\code{full_history = FALSE} defines
the policy $A_k = I\{X_k > 0\}$. Similarly, at stage $k=2$,
\code{function(X_1, X_2) 1*(X_1>0)*(X_2>0)} in combination with
\code{full_history = TRUE} defines the policy $A_2 = I\{X_1>0, X_2>0\}$.
The input \code{replicate = TRUE} will reuse the provided policy functions at
each stage if possible.\\

\code{policy_learn()} specifies a policy learning algorithm which can be used
directly on a policy data object or as input to \code{policy_eval()}.
The \code{type} argument selects the method. Table \ref{tab:policy_learn_type}
provides an overview of the method types, dependencies and limitations.
\begin{table}[htbp]
\centering
\begin{tabular}{ |c|p{4cm}|c|p{4cm}| }
 \hline
 \code{type} \textbf{argument}  & \textbf{Method} & \textbf{Imports}            & \textbf{Limitations} \\
 \hline
 \code{"ql"}  & $Q$-learning &                    & \\
 \hline
 \code{"drql"} & Doubly Robust $Q$-learning. Algorithm \ref{alg:dr_QV_learning}, \ref{alg:dr_RQV_learning}. &                    & \\
 \hline
 \code{"ptl"}  & Policy tree learning. Algorithm \ref{alg:sequential_value_search}. & \pkg{policytree}   & {Realistic policy learning implemented for dichotomous action sets.}\\
 \hline
 \code{"owl"}  & Outcome weighted learning & \pkg{DTRlearn2}    & {No realistic policy learning. Fixed number of stages. Dichotomous action set. Augmentation terms are not cross-fitted.} \\
 \hline
 \code{"earl"} & Efficient augmented and relaxation learning & \pkg{DynTxRegime}  & {Single stage. No cross-fitting. No realistic policy learning. Dichotomous action set.} \\
 \hline
 \code{"rwl"}  & Residual weighted learning & \pkg{DynTxRegime}  & {Same as  \code{"earl"}.} \\
 \hline
\end{tabular}
\caption{Overview of policy learning methods and their dependencies and limitations.
\label{tab:policy_learn_type}}
\end{table}

A cross-fitted doubly robust $V$-restricted $Q$-learning algorithm, see Algorithm
\ref{alg:dr_QV_learning} and \ref{alg:dr_RQV_learning}, may be specified as follows:

\begin{CodeChunk}
\begin{CodeOutput}
  policy_learn(type = "drql",
               control = list(qv_models = q_glm(~X)),
               full_history = FALSE,
               alpha = 0.05,
               L = 10)
\end{CodeOutput}
\end{CodeChunk}

The control argument \code{qv_models} is a single model or a list of models specifying the \QV-models.
We will subsequently describe these models in detail. Note that a \QV-model is fitted
for each action in the action set. The argument \code{full_history} specifies the history
available to the \QV-models similar to \code{full_history} in \code{policy_def()}.
The argument \code{alpha} is the probability threshold for defining
the set of realistic actions. The default value is \code{alpha = 0}. Finally, the argument
\code{L} is the number of folds used in the cross-fitting procedure.

Similarly, a cross-fitted doubly robust sequential value search procedure based on
decision trees, see Algorithm \ref{alg:sequential_value_search}, may be specified
as follows:

\begin{CodeChunk}
\begin{CodeOutput}
  policy_learn(type = "ptl",
               control = control_ptl(policy_vars = c("X"),
                                     depth,
                                     split.step,
                                     min.node.size,
                                     hybrid,
                                     search.depth)
               full_history = FALSE,
               alpha = 0.05,
               L = 10)
\end{CodeOutput}
\end{CodeChunk}

The function \code{control_ptl()} helps set the default control
arguments for \code{type = "ptl"}. Similar functions are available for every
policy learning type. The control argument \code{policy_vars} is a character vector or a list
of character vectors further subsetting the history available to the decision
tree model. The control arguments \code{depth}, \code{split.step}, \code{min.node.size},
and \code{search.depth} are
directly passed on to \code{policytree::policy_tree()}. Each of these arguments must be an
integer or an integer vector. The control argument \code{hybrid} is a logical value
indicating whether to use \code{policytree::policy_tree()} or
\code{policytree::hybrid_policy_tree()}.

The value of a user-specified policy or a policy learning algorithm can be
estimated using \code{policy_eval()}. The evaluation can be based on inverse
probability weighting or outcome regression. However, the default is to use the doubly
robust value estimator given by Algorithm \ref{alg:cross_fitted_learned_value}:

\begin{CodeChunk}
\begin{CodeOutput}
  policy_eval(type = "dr",
              policy_data,
              policy,
              policy_learn,
              g_models = g_glm(~ X+B),
              g_full_history = FALSE,
              q_models = q_glm(~ A*X),
              q_full_history = FALSE,
              M = 10)
\end{CodeOutput}
\end{CodeChunk}

\code{g_models} and \code{g_full_history} specifies the modelling of the
$g$-functions. If \code{g_full_history = FALSE} and a single $g$-model is
provided, a single Markov type model across all stages is fitted. In this case
a generalized linear model is fitted with a model matrix given the formula \code{~ X+B}.
If \code{g_full_history = TRUE} or \code{g_models} is a list, a $g$-function is
fitted for each stage.
Similarly, \code{q_models} and \code{q_full_history} specifies the modelling of
the $Q$-functions. A model is fitted at each stage. If
\code{q_full_history = FALSE} and a single $Q$-model is provided, the model is
reused at each stage with the same design. Alternatives to \code{g_glm()} and
\code{q_glm()} are listed in Table \ref{g_q_models}.
The models are created to save the design
specifications, which is useful for cross-fitting.
\code{M} is the number of folds in the cross-fitting procedure.

\begin{table}[htbp]
\centering
\begin{tabular}{ |c|p{4cm}|c|p{4cm}| }
 \hline
 \textbf{Call}  & \textbf{Method} & \textbf{Imports} & \textbf{Limitations} \\
 \hline
 \code{g_empir} & empirical (conditional) probabilities &  & \\
 \hline
 \code{g_glm/q_glm} & generalized linear model & \pkg{stats} & {\code{g_glm}: dichotomous actions}\\
 \hline
 \code{g_glmnet/q_glmnet}  & lasso and elastic-net regularized generalized linear models & \pkg{glmnet}   & {\code{g_glmnet}: dichotomous actions}\\
 \hline
 \code{g_rf/q_rf}  & random forests & \pkg{ranger} & {} \\
 \hline
 \code{g_sl/q_sl} & super learner prediction algorithm & \pkg{SuperLearner}  & {\code{g_sl}: dichotomous actions} \\
 \hline
\end{tabular}
\caption{Overview of available $g$-model and $Q$-model constructors.
\label{g_q_models}}
\end{table}


\FloatBarrier

\section{Examples}\label{sec:examples}

In this section we go through two reproducible examples based on
simulated data sets that illustrates various applications of the \pkg{polle}
package.

In Section \ref{sec:ex:singlestage} we consider a single-stage problem.
We demonstrate how the policy evaluation framework handles static policies to
obtain estimates of causal effects. Furthermore, we evaluate the true optimal
dynamic policy using highly adaptive nuisance models and use doubly robust
blip learning to obtain an estimate of the same optimal policy.
In Section \ref{sec:ex:twostage} we study a problem with two fixed stages. We
show how to create a policy data object from raw data on wide format and how to
formulate the optimal dynamic policy over multiple stages. We use $g$-models and
$Q$-models with custom data designs to evaluate the policy and showcase the use
of policy trees.

Additional examples including handling a stochastic number of stages and multiple actions
are available in the package vignettes for \code{policy_data()}, \code{policy_eval()}, and \code{policy_learn()}.



\subsection{Single-stage problem} \label{sec:ex:singlestage}

To illustrate the usage of the \pkg{polle} package we first consider a
single-stage problem. Here we consider data from a simulation where the optimal policy
is known. We consider observed data from the directed acyclic graph (DAG) given in Figure \ref{fig:dag1}.
\begin{figure}[htbp]
  \centering
  \includegraphics[width=0.35\textwidth]{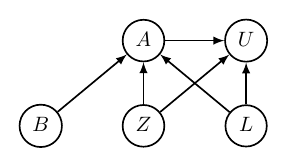}
  \caption[]{\label{fig:dag1}
    Single-stage problem with treatment variable \(A\), utility \(U\), and
    confounders \(B, Z, L\).
  }
\end{figure}

The utility/reward/response is in this example defined as the conditional Gaussian
distribution
\begin{align*}
  U\mid Z,L,A \sim \mathcal{N}(Z+L+A\cdot\{\gamma Z + \alpha L + \beta\}, \sigma^{2})
\end{align*}
with independent state covariates/variables \(Z,L\sim\operatorname{Uniform}([0,1])\), and treatment, \(A\),
defined by the logistic regression model
\begin{align*}
A\mid Z,L,B \sim \operatorname{Bernoulli}(\operatorname{expit}\{\kappa Z^{-2}(Z+L-1) + \delta B\})
\end{align*}
where \(B\sim\operatorname{Bernoulli}(\pi)\) is an additional independent state covariate, and
\(\operatorname{expit}\) is the inverse logistic link function. Here we consider
the choices \(\pi = 0.3, \kappa=0.1, \Delta=0.5, \alpha=1, \beta=-2.5, \gamma=3, \sigma=1\):
\begin{Schunk}
\begin{Sinput}
R> library("polle")
R> par0 <- c(p = .3, k = .1,  d = .5, a = 1, b = -2.5, c = 3)
R> d <- sim_single_stage(n = 5e2, seed=1, par=par0)
R> head(d, 4)
\end{Sinput}
\begin{Soutput}
           Z          L B A          U
1  1.2879704 -1.4795962 0 1 -0.9337648
2  1.6184181  1.2966436 0 1  6.7506026
3  1.2710352 -1.0431352 0 1 -0.3377580
4 -0.2157605  0.1198224 1 0  1.4993427
\end{Soutput}
\end{Schunk}

The data is first transformed using \code{policy_data()} with
instructions on which variables defines the \emph{action}, state \emph{covariates} and
the \emph{utility}:
\begin{Schunk}
\begin{Sinput}
R> pd <- policy_data(d, action="A", covariates=list("Z", "B", "L"), utility="U")
R> pd
\end{Sinput}
\begin{Soutput}
Policy data with n = 500 observations and maximal K = 1 stages.

     action
stage   0   1   n
    1 278 222 500

Baseline covariates:
State covariates: Z, B, L
Average utility: -0.98
\end{Soutput}
\end{Schunk}

\subsubsection{Policy Evaluation}

A single-stage \emph{policy} is mapping from the history $H = (B, Z, L)$ onto
the set of actions $\mathcal{A} = \{0,1\}$.
It is possible to evaluate both user-defined policies as well as learning a
policy from the data using \pkg{polle}. Here we first illustrate how to estimate
the value of a \emph{static policy} where all individuals are given action `1`
irrespective of their covariate values. Policies are defined using
\code{policy_def()} which expects a function as
input or, as here, a numeric vector specifying the static policy:

\begin{Schunk}
\begin{Sinput}
R> p1 <- policy_def(1, name="A=1")
R> p1
\end{Sinput}
\begin{Soutput}
Policy with argument(s)
policy_data
\end{Soutput}
\end{Schunk}

The policy can be applied to a \code{policy_data} object to get the individual actions:
\begin{Schunk}
\begin{Sinput}
R> p1(pd) |> head(3)
\end{Sinput}
\begin{Soutput}
Key: <id, stage>
      id stage      d
   <int> <int> <char>
1:     1     1      1
2:     2     1      1
3:     3     1      1
\end{Soutput}
\end{Schunk}

Note that a \code{policy} applied to a \code{policy_data} object returns a \code{data.table}
with keys \code{id} and \code{stage}. In this case \code{policy_data()} has set the default \code{id} values.

The value of the policy can then be estimated using \code{policy_eval()}:
\begin{Schunk}
\begin{Sinput}
R> (pe1 <- policy_eval(pd, policy=p1))
\end{Sinput}
\begin{Soutput}
    Estimate Std.Err   2.5% 97.5%   P-value
A=1   -2.674  0.2116 -3.089 -2.26 1.331e-36
\end{Soutput}
\end{Schunk}

This provides an estimate of the average potential outcome,
\(\E[U^{(a=1)}]\). By default, a doubly robust estimator given by
Algorithm \ref{alg:cross_fitted_value} without cross-fitting is used
to estimate the value. A logistic regression model with all main effects is
used to model the \(g\)-function and a linear regression model
with all interaction effects between the action and each of the state
covariates is used to model the \(Q\)-function. We will later
revisit how to estimate the value of the policy using flexible machine
learning models and cross-fitting.

In the same way, we can estimate the value under the action `0`:
\begin{Schunk}
\begin{Sinput}
R> (pe0 <- policy_eval(pd, policy=policy_def(0, name = "A=0")))
\end{Sinput}
\begin{Soutput}
    Estimate Std.Err    2.5%  97.5% P-value
A=0 -0.02243 0.08326 -0.1856 0.1408  0.7877
\end{Soutput}
\end{Schunk}

Finally, the average treatment effect,
\(ATE := \E\{U^{(a=1)}) -  U^{(a=0)}\}\), can then be estimated as:
\begin{Schunk}
\begin{Sinput}
R> estimate(merge(pe0, pe1), function(x) x[2]-x[1], labels="ATE")
\end{Sinput}
\begin{Soutput}
    Estimate Std.Err   2.5%  97.5%   P-value
ATE   -2.652  0.1737 -2.992 -2.312 1.236e-52
\end{Soutput}
\end{Schunk}

The function \code{lava::merge.estimate()} combines the fitted influence curve for each estimate \footnote{https://cran.r-project.org/web/packages/lava/vignettes/influencefunction.html}.
The influence curve matrix is available via \code{IC()}:

\begin{Schunk}
\begin{Sinput}
R> IC(merge(pe0, pe1)) |> head(3)
\end{Sinput}
\begin{Soutput}
          A=0       A=1
1 -0.08832404 0.6568576
2  2.61912976 9.6387445
3  0.27539152 1.0710632
\end{Soutput}
\end{Schunk}

The standard errors for the transformation \(f(x_1, x_2) = x_2 - x_1\) is then
given by the delta method.

In this case, we know that the optimal decision boundary is defined by the
hyper-plane \(\gamma Z + \alpha L + \beta = 0\). Again, we use \code{policy_def()}
to define the optimal policy:

\begin{Schunk}
\begin{Sinput}
R> p_opt <- policy_def(
+    function(Z, L) 1*((par0["c"]*Z + par0["a"]*L + par0["b"])>0),
+    name="optimal")
\end{Sinput}
\end{Schunk}

We estimate the value of the optimal policy using Algorithm
\ref{alg:cross_fitted_value}. Specifically, we use \code{M} fold cross-fitting and
super learners for the \(g\)-function and \(Q\)-function including random forests
regression and generalized additive models as implemented in the \pkg{SuperLearner}
package \citep{superlearner-package}:

\begin{Schunk}
\begin{Sinput}
R> set.seed(1)
R> policy_eval(
+    pd,
+    policy = p_opt,
+    g_models = g_sl(SL.library = c("SL.glm", "SL.ranger", "SL.gam")),
+    q_models = q_sl(SL.library = c("SL.glm", "SL.ranger", "SL.gam")),
+    M = 5
+  )
\end{Sinput}
\begin{Soutput}
        Estimate Std.Err   2.5%  97.5%   P-value
optimal    0.373  0.1109 0.1556 0.5903 0.0007693
\end{Soutput}
\end{Schunk}

\subsubsection{Policy learning}

In most applications the optimal policy is of course not known. Instead we seek to
estimate/learn the optimal policy from the data. The function
\code{policy_learn()} constructs a policy learner. Here we specify a
cross-fitted doubly robust blip learning algorithm almost identical to
Algorithm \ref{alg:dr_QV_learning}:

\begin{Schunk}
\begin{Sinput}
R> pl <- policy_learn(
+    type = "blip",
+    L = 5,
+    control = control_blip(blip_models = q_glm(formula = ~ Z + L))
+  )
\end{Sinput}
\end{Schunk}

The policy learner is restricted to $V=(Z,L)$ given by the \code{formula} argument.
Remember that \code{L} is the number of cross-fitting folds. The algorithm can
be applied directly resulting in a policy object:

\begin{Schunk}
\begin{Sinput}
R> set.seed(1)
R> po <- pl(
+    pd,
+    g_models = g_sl(SL.library = c("SL.glm", "SL.ranger", "SL.gam")),
+    q_models = q_sl(SL.library = c("SL.glm", "SL.ranger", "SL.gam"))
+  )
R> po
\end{Sinput}
\begin{Soutput}
Policy object with list elements:
blip_functions, q_functions_cf, g_functions_cf, action_set,
stage_action_sets, alpha, K, folds
Use 'get_policy' to get the associated policy.
\end{Soutput}
\end{Schunk}

The actions of the learned policy are available through \code{get_policy()} returning the associated \code{policy}
function, which we can apply to a \code{policy_data} object:

\begin{Schunk}
\begin{Sinput}
R> get_policy(po)(pd) |> head(3)
\end{Sinput}
\begin{Soutput}
Key: <id, stage>
      id stage      d
   <int> <int> <char>
1:     1     1      1
2:     2     1      1
3:     3     1      1
\end{Soutput}
\end{Schunk}

Alternatively, we can use \code{get_policy_functions()} to return the policy as a function of a \code{data.table}:

\begin{Schunk}
\begin{Sinput}
R> get_policy_functions(po, stage = 1)(H = data.table(Z = c(1,-2), L = c(0.5, -1)))
\end{Sinput}
\begin{Soutput}
[1] "1" "0"
\end{Soutput}
\end{Schunk}

The value of the learned policy can also be estimated directly via
\code{policy_eval()}, see Algorithm \ref{alg:cross_fitted_learned_value}:

\begin{Schunk}
\begin{Sinput}
R> set.seed(1)
R> pe <- policy_eval(
+    pd,
+    policy_learn = pl,
+    g_models = g_sl(SL.library = c("SL.glm", "SL.ranger", "SL.gam")),
+    q_models = q_sl(SL.library = c("SL.glm", "SL.ranger", "SL.gam")),
+    M = 5
+  )
R> pe
\end{Sinput}
\begin{Soutput}
     Estimate Std.Err   2.5%  97.5%   P-value
blip     0.39  0.1109 0.1727 0.6073 0.0004359
\end{Soutput}
\end{Schunk}

Note that the cross-fitting procedure is nested in this case, i.e.,
\(M \times L \) \(g\)-functions and \(Q\)-functions are fitted. The resulting
policy actions are displayed in Figure \ref{fig:ex1_dplot} along with the true
optimal decision boundary.

\begin{figure}[!ht]
\centering
\begin{Schunk}

\includegraphics[width=.75\textwidth]{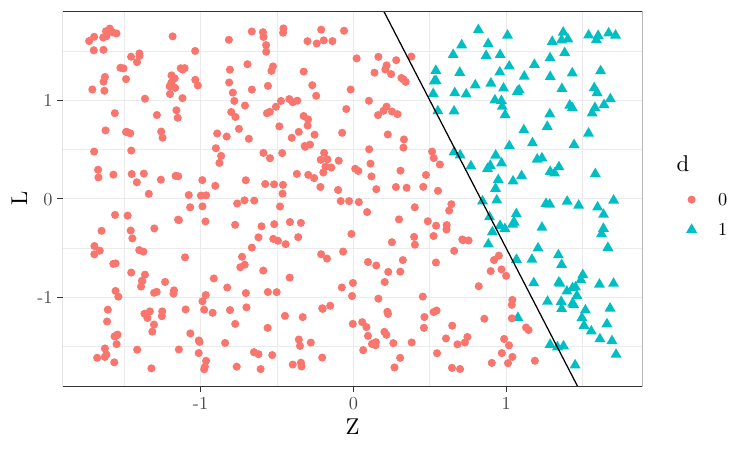} \end{Schunk}
\caption{Fitted policy actions based on doubly robust blip
learning. The black line shows the true optimal decision boundary.}
\label{fig:ex1_dplot}
\end{figure}

\subsection{Two-stage problem} \label{sec:ex:twostage}

In this example we consider a two-stage problem. An observation can be written as $O := (S_1, A_1, S_2, A_2, S_3)$, where $S_1 = (C_1, L_1, U_1)$, $S_2 = (C_2, L_2, U_2)$, and $S_3 = (L_3, U_3)$. The state covariates (cost $C_k$ and load $L_k$) and action variables ($A_k$) are associated with the DAG in Figure \ref{fig:dag2}.

\begin{figure}[htbp]
  \centering
  \includegraphics[width=0.7\textwidth]{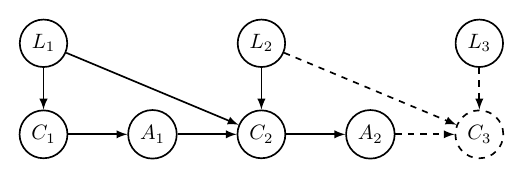}
  \caption[]{
    Two-stage problem with states \((L_{1},C_{1}), (L_{2},C_{2}), L_{3}\) with
    exogenous component  \(L_{k}, k=1,2,3\). As the utility in this example does not depend on the last
    endogenous state, \(C_{3}\), it can be omitted from the analysis. \label{fig:dag2}}
\end{figure}

Specifically, the costs, loads and actions are given by the structural model
\begin{align*}
L_{1} \sim& \mathcal{N}(0, 1)\\
C_{1} \mid L_{1} \sim& \mathcal{N}(L_1, 1)\\
A_1 \mid C_1 \sim& \text{Bernoulli}(\text{expit}(\beta C_1))\\
L_{2} \sim& \mathcal{N} (0, 1)\\
C_{2} \mid A_1, L_1 \sim& \mathcal{N}(\gamma L_1 + A_1, 1)\\
A_2 \mid C_2 \sim& \text{Bernoulli}(\text{expit}(\beta C_2))\\
L_{3} \sim& \mathcal{N} (0, 1)
\end{align*}
for parameters $\gamma, \beta \in \mathbb{R}$. The rewards are given by
\begin{align*}
U_1 &= L_1\\
U_2 &= A_1\cdot C_1 + L_2 \\
U_3 &= A_2\cdot C_2 + L_3
\end{align*}
Remember that the utility is the sum of the rewards, i.e., $U = U_1 + U_2 + U_3$.
In this problem we consider the parameter choices \(\gamma = 0.5, \beta = 1\):

\begin{Schunk}
\begin{Sinput}
R> par0 <- c(gamma = 0.5, beta = 1)
R> d <- sim_two_stage(2e3, seed=1, par=par0)
R> head(d[, -(1:2)], 3)
\end{Sinput}
\begin{Soutput}
          L_1       C_1   A_1        L_2       C_2   A_2         L_3
        <num>     <num> <int>      <num>     <num> <int>       <num>
1:  0.9696772  1.711279     1 -0.7393434  2.424370     1 -0.83124340
2: -2.1994065 -2.643124     0  0.4828756 -2.664728     0 -0.07151015
3:  1.9480938  2.061934     0  0.4803055  2.474761     1  0.40785209
          U_1       U_2         U_3
        <num>     <num>       <num>
1:  0.9696772 0.9719356  1.59312684
2: -2.1994065 0.4828756 -0.07151015
3:  1.9480938 0.4803055  2.88261357
\end{Soutput}
\end{Schunk}

The data is transformed using \code{policy_data()} with instructions on which
variables define the \emph{actions}, \emph{covariates} and the \emph{rewards}
at each stage.
\begin{Schunk}
\begin{Sinput}
R> pd <- policy_data(d,
+                    action = c("A_1", "A_2"),
+                    covariates = list(L = c("L_1", "L_2"),
+                                      C = c("C_1", "C_2")),
+                    utility = c("U_1", "U_2", "U_3"))
R> pd
\end{Sinput}
\begin{Soutput}
Policy data with n = 2000 observations and maximal K = 2 stages.

     action
stage    0    1    n
    1 1017  983 2000
    2  819 1181 2000

Baseline covariates:
State covariates: L, C
Average utility: 0.84
\end{Soutput}
\end{Schunk}

\subsubsection{Policy Evaluation}

The optimal policy $d_0 = (d_{0,1}, d_{0,2})$ is identified via the $Q$-functions. At stage $2$, the $Q$-function is given by
\begin{align*}
Q_{0,2}(h_2,a_2) =&E[U\mid H_2 = h_2, A_2 = a_2]\\
=& l_1 + a_1 c_1 + l_2 + a_2 c_2.
\end{align*}
Thus, the optimal policy at stage 2 is
\begin{align*}
d_{0,2}(h_2) =& \text{arg}\max_{a_2 \in \{0,1\}} Q_2(h_2,a_2)\\
=& I\{c_2 >0\}.
\end{align*}
At stage $1$, the $Q$-function under the optimal policy at stage 2 is given by
\begin{align*}
Q^{d_0}_{0,1}(h_1,a_1) =&E[Q_{0,2}(H_2, d_{0,2}(h_2))\mid H_1 = h_1, A_a = a_1]\\
=& l_1 + a_1 c_1 + E[I\{C_2 >0\}C_2 \mid L_1 = l_1, A_1 = a_1].
\end{align*}
Let
\begin{align*}
\kappa(a_1, l_1) =&  E[I\{C_2 >0\}C_2 \mid L_1 = l_1, A_1 = a_1]\\
=& E[I\{C_2 >0\} \mid L_1 = l_1, A_1 = a_1] E[C_2 \mid L_1 = l_1, A_1 = a_1, C_2 > 0]\\
=& \left(1 - \Phi(-\{\gamma l_1 + a_1\})\right) \left(\gamma l_1 + a_1 + \frac{\phi(-\{\gamma l_1 + a_1\})}{1 - \Phi(-\{\gamma l_1 + a_1\})}\right).
\end{align*}
The optimal policy at stage 1 can now be written as
\begin{align*}
d_{0,1}(h_1) =& \text{arg}\max_{a_1 \in \{0,1\}} Q_1(h_1,a_1)\\
=& I\{(c_1 + \kappa(1, l_1) - \kappa(0, l_1)) >0\}.
\end{align*}

The basis for defining the optimal policy are the histories $H_1 = (L_1, C_1)$
and $H_2 = (L_1, C_1, A_1, L_2, C_2)$ which are available via
\code{get_history()}:

\begin{Schunk}
\begin{Sinput}
R> get_history(pd, stage = 1, full_history = TRUE)$H |> head(3)
\end{Sinput}
\begin{Soutput}
Key: <id, stage>
      id stage        L_1       C_1
   <int> <num>      <num>     <num>
1:     1     1  0.9696772  1.711279
2:     2     1 -2.1994065 -2.643124
3:     3     1  1.9480938  2.061934
\end{Soutput}
\begin{Sinput}
R> get_history(pd, stage = 2, full_history = TRUE)$H |> head(3)
\end{Sinput}
\begin{Soutput}
Key: <id, stage>
      id stage    A_1        L_1        L_2       C_1       C_2
   <int> <num> <char>      <num>      <num>     <num>     <num>
1:     1     2      1  0.9696772 -0.7393434  1.711279  2.424370
2:     2     2      0 -2.1994065  0.4828756 -2.643124 -2.664728
3:     3     2      0  1.9480938  0.4803055  2.061934  2.474761
\end{Soutput}
\end{Schunk}

We use the \code{policy_def()} function to define the optimal policy:

\begin{Schunk}
\begin{Sinput}
R> kappa <- function(mu){
+    pnorm(q = -mu, lower.tail = FALSE) *
+      (mu + dnorm(-mu) / pnorm(-mu, lower.tail = FALSE))
+  }
R> p_opt <- policy_def(
+    list(function(C_1, L_1){
+      1*((C_1 +
+         kappa(par0[["gamma"]] * L_1 + 1) -
+         kappa(par0[["gamma"]] * L_1)
+      ) > 0)
+    },
+    function(C_2){
+      1*(C_2 > 0)
+    }),
+    full_history = TRUE,
+    name = "optimal"
+  )
R> p_opt
\end{Sinput}
\begin{Soutput}
Policy with argument(s)
policy_data
\end{Soutput}
\end{Schunk}

The optimal policy can be applied directly on the policy data:
\begin{Schunk}
\begin{Sinput}
R> p_opt(pd) |> head(3)
\end{Sinput}
\begin{Soutput}
Key: <id, stage>
      id stage      d
   <int> <int> <char>
1:     1     1      1
2:     1     2      1
3:     2     1      0
\end{Soutput}
\end{Schunk}

Doubly robust evaluation of the optimal policy requires modelling the
$g$-functions and $Q$-functions. In this case, the $g$-function is repeated
at each stage. Thus, we may combine $(C_1, A_1)$ and $(C_2, A_2)$ when fitting
the $g$-function. The combined state histories and actions are available
through the \code{get_history()} function with \code{full_history = FALSE}:

\begin{Schunk}
\begin{Sinput}
R> get_history(pd, full_history = FALSE)$H |> head(3)
\end{Sinput}
\begin{Soutput}
Key: <id, stage>
      id stage          L         C
   <int> <int>      <num>     <num>
1:     1     1  0.9696772  1.711279
2:     1     2 -0.7393434  2.424370
3:     2     1 -2.1994065 -2.643124
\end{Soutput}
\begin{Sinput}
R> get_history(pd, full_history = FALSE)$A |> head(3)
\end{Sinput}
\begin{Soutput}
Key: <id, stage>
      id stage      A
   <int> <int> <char>
1:     1     1      1
2:     1     2      1
3:     2     1      0
\end{Soutput}
\end{Schunk}

Similarly, when using \code{policy_eval()}, we can specify the
structure of the used histories:

\begin{Schunk}
\begin{Sinput}
R> pe_opt <- policy_eval(pd,
+                        policy = p_opt,
+                        g_models = g_glm(),
+                        g_full_history = FALSE,
+                        q_models = list(q_glm(), q_glm()),
+                        q_full_history = TRUE)
R> pe_opt
\end{Sinput}
\begin{Soutput}
        Estimate Std.Err  2.5% 97.5%   P-value
optimal    1.311 0.06578 1.182  1.44 2.067e-88
\end{Soutput}
\end{Schunk}

On closer inspection we see that a single $g$-model has been fitted across all stages:
\begin{Schunk}
\begin{Sinput}
R> get_g_functions(pe_opt)
\end{Sinput}
\begin{Soutput}
$all_stages
$model

Call:  NULL

Coefficients:
(Intercept)            L            C  
    0.01591      0.03145      0.98013  

Degrees of Freedom: 3999 Total (i.e. Null);  3997 Residual
Null Deviance:	    5518 
Residual Deviance: 4361 	AIC: 4367


attr(,"full_history")
[1] FALSE
\end{Soutput}
\end{Schunk}

If \code{q_models} is not a list, the provided model is reused at each stage.
In this case the full history is used at both stages:

\begin{Schunk}
\begin{Sinput}
R> get_q_functions(pe_opt)
\end{Sinput}
\begin{Soutput}
$stage_1
$model

Call:  NULL

Coefficients:
(Intercept)           A1          L_1          C_1       A1:L_1  
    0.52415      0.44348      0.17462      0.09043      0.21854  
     A1:C_1  
    0.92899  

Degrees of Freedom: 1999 Total (i.e. Null);  1994 Residual
Null Deviance:	    6029 
Residual Deviance: 2772 	AIC: 6343


$stage_2
$model

Call:  NULL

Coefficients:
(Intercept)           A1         A_11          L_1          L_2  
  -0.014697     0.140420    -0.148007    -0.118571     0.002233  
        C_1          C_2      A1:A_11       A1:L_1       A1:L_2  
   0.124817    -0.031178     0.046876     0.171946     0.023246  
     A1:C_1       A1:C_2  
  -0.113314     0.944778  

Degrees of Freedom: 1999 Total (i.e. Null);  1988 Residual
Null Deviance:	    3580 
Residual Deviance: 1881 	AIC: 5579


attr(,"full_history")
[1] TRUE
\end{Soutput}
\end{Schunk}

Note that in practice only the residual value is used as input to the (residual) $Q$-models, i.e.,
\begin{align*}
Q_{0,2, \text{res}}(h_2,a_2) :=&E[U_3\mid H_2 = h_2, A_2 = a_2] \\
=& a_2 c_2\\
Q^{d_0}_{0,1, \text{res}}(h_1,a_1) :=& E[U_2 + Q_{0,2, \text{res}}(H_2, d_{0,2}(h_2))\mid H_1 = h_1, A_a = a_1]\\
=& a_1 c_1 + E[I\{C_2 >0\}C_2 \mid L_1 = l_1, A_1 = a_1].
\end{align*}

The fitted values of the $g$-functions and $Q$-functions are easily extracted using \code{predict()}:

\begin{Schunk}
\begin{Sinput}
R> predict(get_g_functions(pe_opt), pd) |> head(3)
\end{Sinput}
\begin{Soutput}
Key: <id, stage>
      id stage        g_0        g_1
   <int> <int>      <num>      <num>
1:     1     1 0.15139841 0.84860159
2:     1     2 0.08557919 0.91442081
3:     2     1 0.93363059 0.06636941
\end{Soutput}
\begin{Sinput}
R> predict(get_q_functions(pe_opt), pd) |> head(3)
\end{Sinput}
\begin{Soutput}
Key: <id, stage>
      id stage       Q_0       Q_1
   <int> <int>     <num>     <num>
1:     1     1  1.817896  4.063055
2:     1     2  1.800290  4.233710
3:     2     1 -2.298324 -4.790946
\end{Soutput}
\end{Schunk}

\subsubsection{Policy Learning}

A $V$-restricted policy can be estimated via the \code{policy_learn()} function.
In this case we use recursive doubly robust value search based on the
\pkg{policytree} package, see Algorithm \ref{alg:sequential_value_search}:
\begin{Schunk}
\begin{Sinput}
R> pl <- policy_learn(type = "ptl",
+                     control = control_ptl(policy_vars = c("C", "L")),
+                     full_history = FALSE,
+                     L = 5)
\end{Sinput}
\end{Schunk}

The policy learner is restricted to $V=(C,L)$ given by the \code{policy_vars}
argument. The learner can be applied directly:

\begin{Schunk}
\begin{Sinput}
R> po <- pl(pd,
+           g_models = g_glm(),
+           g_full_history = FALSE,
+           q_models = q_glm(),
+           q_full_history = TRUE)
R> get_policy(po)(pd) |> head(3)
\end{Sinput}
\begin{Soutput}
Key: <id, stage>
      id stage      d
   <int> <int> <char>
1:     1     1      1
2:     1     2      1
3:     2     1      0
\end{Soutput}
\end{Schunk}

Or the value of the policy learning procedure can be estimated directly using \code{policy_eval()}:
\begin{Schunk}
\begin{Sinput}
R> set.seed(1)
R> pe <- policy_eval(pd,
+                    policy_learn = pl,
+                    g_models = g_glm(),
+                    g_full_history = FALSE,
+                    q_models = q_glm())
R> pe
\end{Sinput}
\begin{Soutput}
    Estimate Std.Err  2.5% 97.5%   P-value
ptl    1.385  0.0809 1.226 1.544 1.068e-65
\end{Soutput}
\end{Schunk}

The associated policy objects are also saved for closer inspection, see Figure \ref{fig:ex2_dplot}:
\begin{Schunk}
\begin{Sinput}
R> po <- get_policy_object(pe)
R> po$ptl_objects
\end{Sinput}
\begin{Soutput}
$stage_1
policy_tree object 
Tree depth:  2 
Actions:  1: 0 2: 1 
Variable splits: 
(1) split_variable: C  split_value: -3.14122 
  (2) split_variable: L  split_value: -1.67452 
    (4) * action: 1 
    (5) * action: 2 
  (3) split_variable: C  split_value: -0.274346 
    (6) * action: 1 
    (7) * action: 2 

$stage_2
policy_tree object 
Tree depth:  2 
Actions:  1: 0 2: 1 
Variable splits: 
(1) split_variable: C  split_value: -0.747456 
  (2) split_variable: C  split_value: -0.811175 
    (4) * action: 1 
    (5) * action: 2 
  (3) split_variable: C  split_value: 0.0237423 
    (6) * action: 1 
    (7) * action: 2 
\end{Soutput}
\end{Schunk}

\begin{figure}[!ht]
\centering
\begin{Schunk}

\includegraphics[width=.75\textwidth]{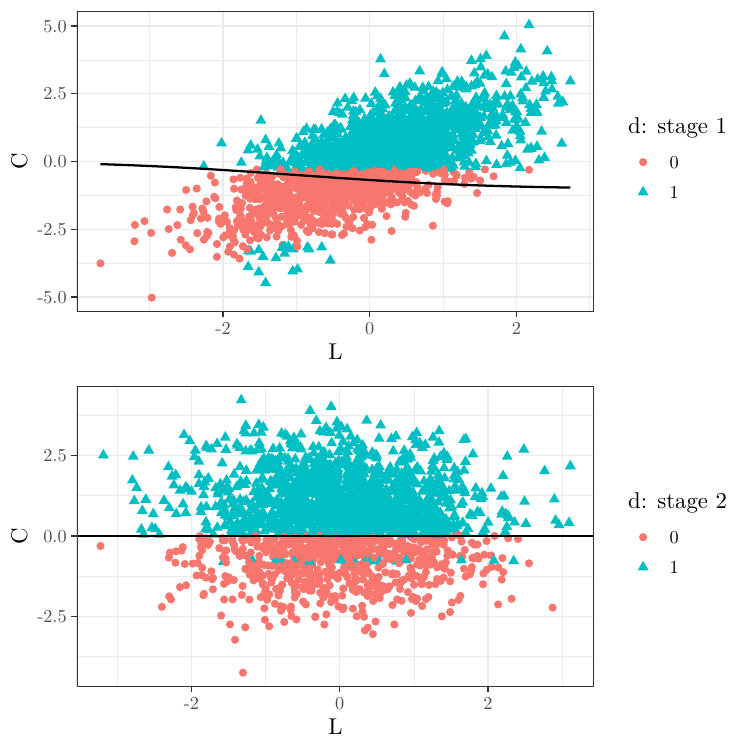} \end{Schunk}
\caption{Fitted policy actions based on policy tree learning. The black lines show the true optimal decision boundaries.}
\label{fig:ex2_dplot}
\end{figure}

\newpage
\section{K-2 Literacy Intervention} \label{sec:harvard_example}

In this section, we demonstrate the application of \pkg{polle} on a real data
example.
To ensure reproducibility we work with a data set publicly available in the
harvard dataverse \url{https://dataverse.harvard.edu/}. Specifically, we use the
kindergarten to second-grade literacy intervention data set \citep{DVN_AVW6KB_2019}
, funded by
the Chan Zuckerberg initiative, which is documented and analyzed in
\citep{kim2019using}. The following analysis is conducted independently of the original
study, and we take full responsibility for any misinterpretations or errors.

The data set contains records of 273 students from kindergarten to second grade
associated with 16 teachers.
The study seeks to investigate the treatment effect of assigning two
types of print texts (10 texts/books in each group) for the students to read
over the summer break. All students received training before the summer break,
and the investigators used a mobile app to engage and monitor each student.
At stage 1 each classroom associated with a given teacher was randomly assigned to
read either conceptually coherent texts (CCT) or leveled text
(LT). At an intermediate time point during the summer break (stage 2), if a
student had completed at least one activity on the mobile app, the student was
classified as a responder and no further actions were initiated. However, all
non-responders were subject to gamification of the app and the parents
were randomly selected to receive text messages with reminders, information, and
encouragement.
The main outcome of the study was the reading measure of academic progress
Rasch unit (MAP RIT) score. The score was measured before and after the summer
break (referred to as spring and fall). Of the 273 students enrolled in the study,
56 students have missing MAP RIT scores. As in the original study, we conduct
a complete case analysis.

We start by loading the data and conducting some basic transformations, which we
include for reproducibility. Note that we define the utility as the difference
in the spring and fall MAP RIT scores. A description of the variables can be
found in Table \ref{tab:k2_var_descriptions}.

\begin{Schunk}
\begin{Sinput}
R> library("readstata13")
R> d <- readstata13::read.dta13("k2smart_public.dta")
R> d <- transform(
+    d,
+    utility = as.numeric(fa_maprit) - as.numeric(sp_maprit),
+    cct = as.logical(cct),
+    responder = as.logical(responder),
+    text = as.logical(text),
+    maprit = as.numeric(sp_maprit),
+    teacher = as.character(t_id_public),
+    attend = as.numeric(sp_pctattend_yr),
+    trc = as.numeric(sp_trcbook_num),
+    dib = as.numeric(sp_dib_score),
+    ell = as.logical(ell),
+    iep = as.logical(iep),
+    grade = as.character(grade_final),
+    male = as.logical(male),
+    familynight = as.logical(familynight)
+  )
\end{Sinput}
\end{Schunk}

\begin{table}[htbp]
\centering
\begin{tabular}{|p{2.5cm}|p{2cm}|p{6cm}|}
 \hline
 \textbf{Variable} & \textbf{Type}  & \textbf{Description} \\
 \hline
 \code{utility} & Numeric & Difference between the spring and fall MAP RIT score \\
 \hline
 \code{cct} & Logical & If \code{TRUE}, the student receives conceptually coherent texts (CCT).
 If \code{FALSE}, the student receives leveled text (LT). \\
 \hline
 \code{responder} & Logical & Respondence indicator \\
 \hline
 \code{text} & Logical & If \code{TRUE} the parents received text messages. \\
 \hline
 \code{maprit} & Numeric & Spring MAP RIT score. \\
 \hline
 \code{teacher} & Character string & Teacher ID. \\
 \hline
 \code{attend} & Numeric & Student attendance percentage. \\
 \hline
 \code{trc} & Numeric & Spring text reading comprehension score. \\
 \hline
 \code{dib} & Numeric & Spring dynamic indicators of basic early literacy skills score. \\
 \hline
 \code{ell} & Logical & If \code{TRUE}, the student received English-language learner services. \\
 \hline
 \code{iep} & Logical & If \code{TRUE}, the student has an individualized education plan. \\
 \hline
 \code{grade} & Character string & Student grade: kindergarten (0), first grade (1), second grade (2). \\
 \hline
 \code{male} & Logical & If \code{TRUE}, the student is a male. \\
 \hline
 \code{familynight} & Logical & If \code{TRUE}, the family of the student attended family night. \\
 \hline
\end{tabular}
\caption{Variable descriptions.
\label{tab:k2_var_descriptions}}
\end{table}

As described we make a complete case analysis:
\begin{Schunk}
\begin{Sinput}
R> d <- subset(d,!is.na(utility))
\end{Sinput}
\end{Schunk}

Finally, we create a stage 1 and stage 2 treatment variable:
\begin{Schunk}
\begin{Sinput}
R> d <- transform(
+    d,
+    A_1 = ifelse(cct, "cct", "lt"),
+    A_2 = ifelse(responder, "continue", ifelse(text, "text", "notext"))
+  )
\end{Sinput}
\end{Schunk}

The \code{policy_data()} function is used to create a policy data object. Note
that \code{responder} is a stage 2 state covariate:
\begin{Schunk}
\begin{Sinput}
R> pd <- policy_data(
+    d,
+    action = c("A_1","A_2"),
+    utility = "utility",
+    baseline=c("maprit",
+               "male",
+               "ell",
+               "iep",
+               "attend",
+               "trc",
+               "dib",
+               "grade",
+               "familynight"),
+    covariates = list(responder = c(NA, "responder"))
+  )
R> print(pd)
\end{Sinput}
\begin{Soutput}
Policy data with n = 217 observations and maximal K = 2 stages.

     action
stage cct continue  lt notext text   n
    1 112        0 105      0    0 217
    2   0       36   0     86   95 217

Baseline covariates: maprit, male, ell, iep, attend, trc,
dib, grade, familynight
State covariates: responder
Average utility: 2.7
\end{Soutput}
\end{Schunk}

Since student responders are not randomized at stage 2, only 4 static realistic
policies exist:
\begin{Schunk}
\begin{Sinput}
R> actions <- list(
+    c("cct","text"),
+    c("cct","notext"),
+    c("lt","text"),
+    c("lt","notext")
+  )
R> static_policies <- lapply(
+    actions,
+    function(a){
+      policy_def(list(
+        function(...) a[1],
+        function(responder) ifelse(responder, "continue", a[2])
+      ),
+      name = paste(a, collapse = "_"))
+    }
+  )
R> head(static_policies[[1]](pd), 4)
\end{Sinput}
\begin{Soutput}
Key: <id, stage>
      id stage      d
   <int> <int> <char>
1:     1     1    cct
2:     1     2   text
3:     2     1    cct
4:     2     2   text
\end{Soutput}
\end{Schunk}

We begin the analysis by comparing the policy value for each of the 4 static
policies. First, we consider a basic inverse probability weighting estimator:
\begin{Schunk}
\begin{Sinput}
R> gm <- list(g_empir(~1),
+             g_empir(~responder))
R> 
R> pe_static_policies_ipw <- lapply(
+    static_policies,
+    function(p){
+      policy_eval(pd,
+                  policy = p,
+                  g_models = gm,
+                  type = "ipw")
+    }
+  )
R> do.call("merge", pe_static_policies_ipw)
\end{Sinput}
\begin{Soutput}
           Estimate Std.Err     2.5% 97.5%  P-value
cct_text      2.920   1.159  0.64838 5.193 0.011760
----------                                         
cct_notext    1.737   1.109 -0.43707 3.912 0.117353
----------                                         
lt_text       3.573   1.120  1.37793 5.769 0.001422
----------                                         
lt_notext     2.504   1.322 -0.08728 5.095 0.058233
\end{Soutput}
\end{Schunk}

Note that the reported standard errors are valid because the $g$-models are known
in a randomized trial. The $g$-models specified by the function \code{g_empir()}
computes the (conditional) empirical probabilities and match them to each student:

\begin{Schunk}
\begin{Sinput}
R> print(get_g_functions(pe_static_policies_ipw[[1]]))
\end{Sinput}
\begin{Soutput}
$stage_1
$tab
        A empir_prob
   <char>      <num>
1:    cct   0.516129
2:     lt   0.483871

$v
character(0)


$stage_2
$tab
Index: <A>
          A responder empir_prob
     <char>    <lgcl>      <num>
1: continue      TRUE  1.0000000
2:   notext     FALSE  0.4751381
3:     text     FALSE  0.5248619

$v
[1] "responder"


attr(,"full_history")
[1] FALSE
\end{Soutput}
\begin{Sinput}
R> predict(get_g_functions(pe_static_policies_ipw[[1]]), pd)[c(1,2,43,44),]
\end{Sinput}
\begin{Soutput}
Key: <id, stage>
      id stage    g_cct g_continue     g_lt  g_notext    g_text
   <int> <int>    <num>      <num>    <num>     <num>     <num>
1:     1     1 0.516129          0 0.483871 0.0000000 0.0000000
2:     1     2 0.000000          0 0.000000 0.4751381 0.5248619
3:    22     1 0.516129          0 0.483871 0.0000000 0.0000000
4:    22     2 0.000000          1 0.000000 0.0000000 0.0000000
\end{Soutput}
\end{Schunk}

Efficiency of the policy value estimates can be increased by using the
doubly robust value scores. As we are fitting $25$ sets of nuisance models we
parallelize the computations via the \pkg{future.apply} package. The variable
names used to specify the nuisance model are available via
\code{get_history_names()}:

\begin{Schunk}
\begin{Sinput}
R> get_history_names(pd, stage = 1)
\end{Sinput}
\begin{Soutput}
 [1] "responder_1" "maprit"      "male"        "ell"        
 [5] "iep"         "attend"      "trc"         "dib"        
 [9] "grade"       "familynight"
\end{Soutput}
\begin{Sinput}
R> get_history_names(pd, stage = 2)
\end{Sinput}
\begin{Soutput}
 [1] "A_1"         "responder_1" "responder_2" "maprit"     
 [5] "male"        "ell"         "iep"         "attend"     
 [9] "trc"         "dib"         "grade"       "familynight"
\end{Soutput}
\end{Schunk}

With this help we can easily specify the $Q$-models using the \pkg{SuperLearner}
package and plug them into the \code{policy_eval()} function:

\begin{Schunk}
\begin{Sinput}
R> sl_lib <- c("SL.mean",
+              "SL.glm",
+              "SL.gam",
+              "SL.ranger",
+              "SL.nnet")
R> qm <- list(
+    q_sl(formula = ~.-responder_1, SL.library = sl_lib),
+    q_sl(formula = ~.-responder_1-responder_2, SL.library = sl_lib)
+  )
\end{Sinput}
\end{Schunk}

\begin{Schunk}
\begin{Sinput}
R> library("future.apply")
R> plan(list(
+    tweak("multisession", workers = 4)
+  ))
R> pe_static_policies_dr <- lapply(
+    static_policies,
+    function(p){
+      set.seed(1)
+      policy_eval(pd,
+                  policy = p,
+                  g_models = gm,
+                  q_models = qm,
+                  q_full_history = TRUE,
+                  type="dr",
+                  M = 25)
+    }
+  )
R> print(do.call("merge", pe_static_policies_dr))
\end{Sinput}
\end{Schunk}

\begin{Schunk}
\begin{Sinput}
R> plan("sequential")
R> print(do.call("merge", pe_static_policies_dr))
\end{Sinput}
\begin{Soutput}
           Estimate Std.Err    2.5% 97.5%   P-value
cct_text      2.825  0.9137  1.0344 4.616 0.0019874
----------                                         
cct_notext    2.289  1.0123  0.3047 4.273 0.0237595
----------                                         
lt_text       3.554  1.0506  1.4951 5.613 0.0007169
----------                                         
lt_notext     1.946  1.1543 -0.3165 4.208 0.0918351
\end{Soutput}
\end{Schunk}

So far the reported standard errors have been overly optimistic because we
ignored the random teacher/classroom effect.
Luckily, when working with influence curves it is easy to adjust for these
types of dependencies.
When estimating the variance we first sum all of the influence curve terms
related to each teacher. The resulting compounded influence curve terms are then.
independent and the variance is computed in the usual fashion.
This approach is similar to computing clustered standard errors
\citep{liang1986longitudinal}. The method is implemented in the \code{estimate()}
function from the \pkg{lava} package \citep{holst2013linear}:

\begin{Schunk}
\begin{Sinput}
R> library("lava")
R> (est <- estimate(do.call("merge", pe_static_policies_dr), id = d$teacher))
\end{Sinput}
\begin{Soutput}
           Estimate Std.Err    2.5% 97.5%   P-value
cct_text      2.825  1.0255  0.8153 4.835 5.870e-03
cct_notext    2.289  0.8746  0.5746 4.003 8.874e-03
lt_text       3.554  0.6580  2.2645 4.844 6.613e-08
lt_notext     1.946  1.3460 -0.6922 4.584 1.483e-01
\end{Soutput}
\end{Schunk}

As is already evident, none of the static policies are statistically different
in terms of marginal value. For completeness we conduct a chi square test:

\begin{Schunk}
\begin{Sinput}
R> pdiff <- function(n) lava::contr(lapply(seq(n-1), \(x) seq(x, n)))
R> estimate(est, f = pdiff(4))
\end{Sinput}
\begin{Soutput}
                          Estimate Std.Err   2.5%  97.5% P-value
[cct_text] - [cct_notext]   0.5364  0.9023 -1.232 2.3048  0.5522
[cct_text] - [lt_text]     -0.7290  1.0779 -2.842 1.3836  0.4988
[cct_text] - [lt_notext]    0.8794  1.6217 -2.299 4.0578  0.5876
[cct_notext] - [lt_text]   -1.2653  0.9088 -3.046 0.5158  0.1638
[cct_notext] - [lt_no....   0.3430  1.6654 -2.921 3.6072  0.8368
[lt_text] - [lt_notext]     1.6083  1.8288 -1.976 5.1928  0.3792

 Null Hypothesis: 
  [cct_text] - [cct_notext] = 0
  [cct_text] - [lt_text] = 0
  [cct_text] - [lt_notext] = 0
  [cct_notext] - [lt_text] = 0
  [cct_notext] - [lt_notext] = 0
  [lt_text] - [lt_notext] = 0 
 
chisq = 2.1289, df = 3, p-value = 0.5461
\end{Soutput}
\end{Schunk}

The function \code{conditional()} allows the user to easily compute the
conditional policy value estimates based on categorical baseline covariates.
Here we group by the baseline covariate \code{male} for the static CCT text
policy:

\begin{Schunk}
\begin{Sinput}
R> estimate(
+    conditional(pe_static_policies_dr[[1]], pd, "male"),
+    id = d$teacher
+  )
\end{Sinput}
\begin{Soutput}
           Estimate Std.Err   2.5% 97.5%   P-value
male:FALSE    1.875  0.6525 0.5961 3.154 4.060e-03
male:TRUE     3.958  0.6706 2.6435 5.272 3.598e-09
\end{Soutput}
\end{Schunk}

Even though none of the static policies have a marginal treatment effect
we may hope to find group specific treatment effects. To investigate further,
we specify a selection of doubly robust $V$-restricted $Q$-learners and estimate
the cross-fitted value of the fitted policies.

We formulate simple linear $\QV$-models using the \code{q_glm()} function as we do not
expect to be able to find complex non-linear treatment associations in this
relatively small data set.
\begin{Schunk}
\begin{Sinput}
R> qvm_formulas <- list(
+    qvm_1 = list(~1, ~A_1),
+    qvm_2 = list(~maprit, ~A_1+maprit),
+    qvm_3 = list(~male, ~A_1+male),
+    qvm_4 = list(~grade, ~A_1+grade),
+    qvm_5 = list(~male+maprit, ~A_1+male+maprit),
+    qvm_6 = list(~male+grade+maprit, ~A_1+grade+male+maprit)
+  )
R> 
R> qvm <- lapply(qvm_formulas,
+         function(form){
+           list(q_glm(form[[1]]), q_glm(form[[2]]))
+         })
\end{Sinput}
\end{Schunk}

The $Q$-models are then passed to the controls of the \code{policy_learn()}
function. Importantly, note that \code{alpha} is set to 0.01 in order to account
for the degenerate structure of the data; A student responder always continue
the treatment in stage 2.
\begin{Schunk}
\begin{Sinput}
R> pl_drql <- mapply(
+    qvm,
+    names(qvm),
+    FUN = function(qv, name){
+      policy_learn(type = "drql",
+                   control = control_drql(qv_models = qv),
+                   full_history = TRUE,
+                   alpha = 0.01,
+                   L = 25,
+                   cross_fit_g_models = FALSE,
+                   name = name)
+    })
\end{Sinput}
\end{Schunk}

The value of the fitted policies are cross-fitted using \code{policy_eval()}:
\begin{Schunk}
\begin{Sinput}
R> plan(list(
+    tweak("multisession", workers = 4)
+  ))
R> set.seed(1)
R> pe_drql <- lapply(
+    pl_drql,
+    function(pl){
+      set.seed(1)
+      policy_eval(pd,
+                  policy_learn = pl,
+                  g_models = gm,
+                  q_models = qm,
+                  q_full_history = TRUE,
+                  type="dr",
+                  M = 25)
+    })
\end{Sinput}
\end{Schunk}
\begin{Schunk}
\begin{Sinput}
R> estimate(do.call(what = "merge", unname(pe_drql)), id = d$teacher)
\end{Sinput}
\begin{Soutput}
      Estimate Std.Err     2.5% 97.5%   P-value
qvm_1    3.063  0.5911  1.90486 4.222 2.188e-07
qvm_2    1.228  0.9686 -0.67020 3.127 2.048e-01
qvm_3    3.382  1.3075  0.81952 5.945 9.689e-03
qvm_4    2.480  1.2339  0.06183 4.898 4.442e-02
qvm_5    1.885  1.0368 -0.14663 3.917 6.898e-02
qvm_6    2.055  1.1366 -0.17234 4.283 7.055e-02
\end{Soutput}
\end{Schunk}

None of the fitted policies show a gain in value compared to the static
policy \code{lt_text}. However, we might still want to study a possible male
treatment interaction further (\code{qvm_3}). We fit policy learner 3 on the
complete data set and summarize the dictated actions:

\begin{Schunk}
\begin{Sinput}
R> set.seed(1)
R> po_drql_male <- pl_drql[["qvm_3"]](pd,
+                                     g_models = gm,
+                                     q_models = qm,
+                                     q_full_history = TRUE)
R> pa_drql_male <- get_policy(po_drql_male)(pd)
R> head(pa_drql_male, 4)
\end{Sinput}
\begin{Soutput}
Key: <id, stage>
      id stage      d
   <int> <int> <char>
1:     1     1    cct
2:     1     2   text
3:     2     1    cct
4:     2     2   text
\end{Soutput}
\end{Schunk}

\begin{Schunk}
\begin{Sinput}
R> pa_drql_male <- merge(pa_drql_male, get_history(pd)$H)
R> pa_drql_male[,.N, list(stage, male,d)][order(stage, male,d)]
\end{Sinput}
\begin{Soutput}
   stage   male        d     N
   <int> <lgcl>   <char> <int>
1:     1  FALSE       lt   118
2:     1   TRUE      cct    99
3:     2  FALSE continue    23
4:     2  FALSE     text    95
5:     2   TRUE continue    13
6:     2   TRUE     text    86
\end{Soutput}
\end{Schunk}

Thus, the fitted policy suggests that males receive CCT and females receive
LT at stage 1 and that all non-responders get text messages at stage 2.

We end this analysis by emphasizing the importance of cross-fitting the policy
learner because it is easy to overfit an optimal policy. We showcase this by fitting
the most complex of the considered policy learners:

\begin{Schunk}
\begin{Sinput}
R> po_drql_6 <- pl_drql[["qvm_6"]](pd,
+                                  g_models = gm,
+                                  q_models = qm,
+                                  q_full_history = TRUE)
\end{Sinput}
\end{Schunk}

The dictated actions are easily plotted using the \code{get_history()} and
\code{get_policy()} functions:

\begin{Schunk}
\begin{Sinput}
R> plot_data <- get_history(pd)$H
R> plot_data <- merge(plot_data,
+                     get_policy(po_drql_6)(pd),
+                     by = c("id", "stage"))
R> library("ggplot2")
R> ggplot(plot_data) +
+    geom_point(aes(x = grade, y = maprit, color = d)) +
+    facet_wrap(~stage+male, labeller = "label_both") +
+    theme_bw()
\end{Sinput}

{\centering \includegraphics{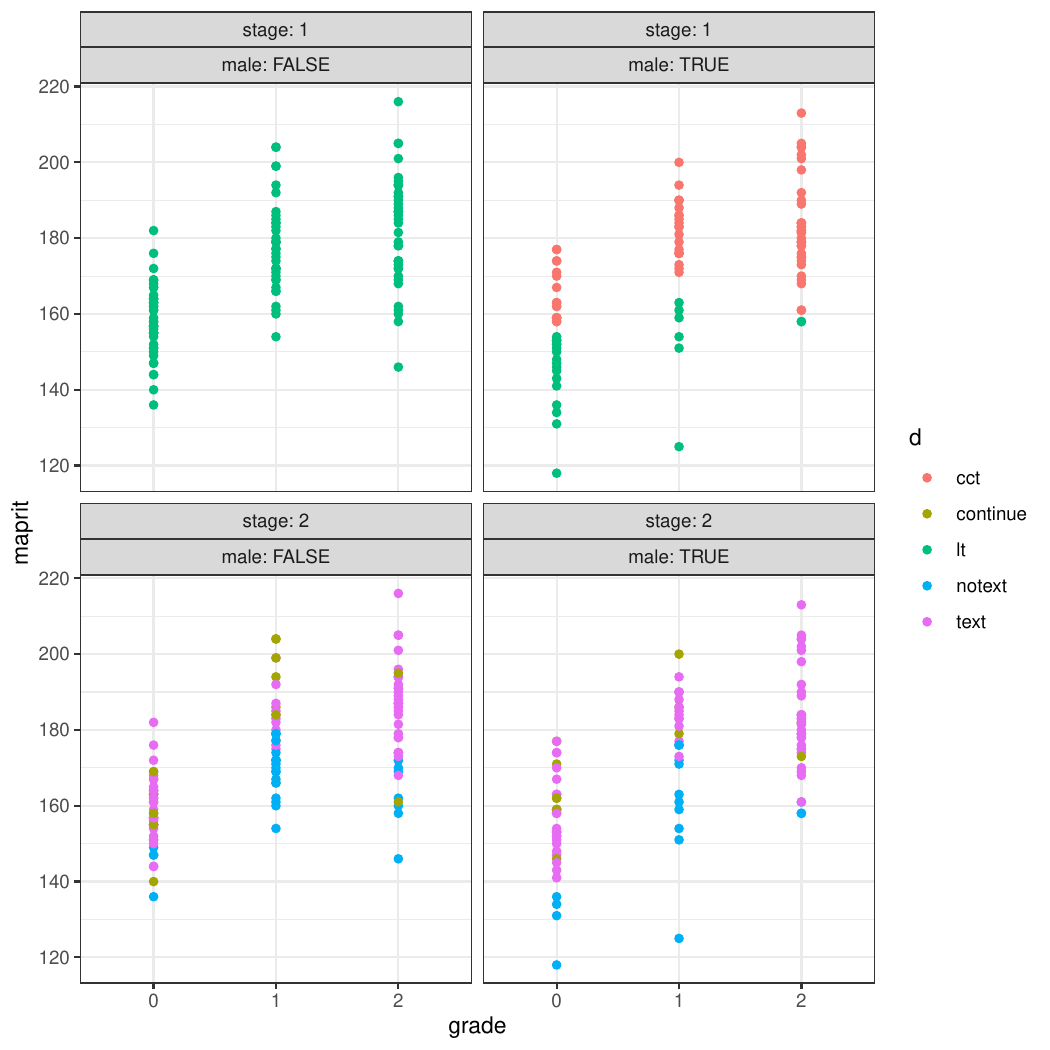} 

}

\end{Schunk}

If we just plug in the resulting fitted policy we get an overly optimistic
estimate of the value.
\begin{Schunk}
\begin{Sinput}
R> plan(list(
+    tweak("multisession", workers = 4)
+  ))
R> set.seed(1)
R> pe_plugin <- policy_eval(pd,
+                           policy = get_policy(po_drql_6),
+                           g_models = gm,
+                           q_models = qm,
+                           q_full_history = TRUE,
+                           type = "dr",
+                           M = 25,
+                           name = "qvm_6_plugin")
R> estimate(pe_plugin + pe_drql[["qvm_6"]], id = d$teacher)
\end{Sinput}
\begin{Soutput}
             Estimate Std.Err    2.5% 97.5%   P-value
qvm_6_plugin    4.912   1.097  2.7625 7.061 7.491e-06
qvm_6           2.055   1.137 -0.1723 4.283 7.055e-02
\end{Soutput}
\end{Schunk}


\section{Summary and discussion} \label{sec:summary}

The \pkg{polle} library is the first unifying \proglang{R} package for learning and
evaluating policies. The package efficiently handles cross-fitting of the
nuisance models and provides protection against (near) positivity violations.
Also, to our knowledge, \pkg{polle} contains the first implementation of doubly
robust restricted $Q$-learning which can serve as sensible benchmark for all
other learning methods.

Of course, \pkg{polle} has its limitations. Future work to be included in the
package includes the handling of missing data and (right) censored observations.
The \code{event} variable included in the policy data object can be extended to
specify missing or censored data similar to that of the \code{Surv} function in
the \pkg{survival} package. Additional models for the censoring distribution
would need to be included. Optimal policies for survival outcomes has been
studied in \citep{bai2017optimal, diaz2018targeted}. See also
\citep{cui2023heter, xu2023treatment,steingrimsson2016doubly}.

In our work we only consider the maximization of a scalar utility value.
However, in some applications a multi-dimensional value vector may more
naturally be of interest. In such cases the set of Pareto efficient policies can
be formulated. An important example would be the task of maximizing the utility
subject to variance constraints in order to learn robust policies. This is
closely related to introducing a penalty term to the loss function, and will be the
subject of future developments to the \pkg{polle} package.

\vspace*{1em}

The package is available directly from the Comprehensive R Archive Network
(CRAN) \citep{cran:polle}. We believe the package will provide practitioners with
much easier access to a broad range of policy learning methods and hope that it also may serve as a
framework for benchmarking as well as implementing new methods for researchers in the policy learning field.
We invite to collaboration on the future development of the package via pull
requests to the github repository
\url{https://github.com/AndreasNordland/polle/}.


\section*{Computational details}

The results in this paper were obtained using \proglang{R}~4.3.2 \citep{rcore} with the
\pkg{polle}~1.4 package. \proglang{R} itself and all packages used are available
from the Comprehensive \proglang{R} Archive Network (CRAN) at
\url{https://CRAN.R-project.org/}.

\FloatBarrier

\section*{Acknowledgments}
This work was supported by a grant from InnovationsFonden Danmark (case no. 8053-00096B).

\bibliography{refs}

\begin{thebibliography}{52}
\newcommand{\enquote}[1]{``#1''}
\providecommand{\natexlab}[1]{#1}
\providecommand{\url}[1]{\texttt{#1}}
\providecommand{\urlprefix}{URL }
\expandafter\ifx\csname urlstyle\endcsname\relax
  \providecommand{\doi}[1]{doi:\discretionary{}{}{}#1}\else
  \providecommand{\doi}{doi:\discretionary{}{}{}\begingroup
  \urlstyle{rm}\Url}\fi
\providecommand{\eprint}[2][]{\url{#2}}

\bibitem[{Athey \emph{et~al.}(2019)Athey, Tibshirani, and
  Wager}]{athey2019generalized}
Athey S, Tibshirani J, Wager S (2019).
\newblock \enquote{Generalized Random Forests.}
\newblock \emph{The Annals of Statistics}, \textbf{47}(2), 1148--1178.

\bibitem[{Athey and Wager(2021)}]{athey2021policy}
Athey S, Wager S (2021).
\newblock \enquote{Policy Learning With Observational Data.}
\newblock \emph{Econometrica}, \textbf{89}(1), 133--161.

\bibitem[{Bai \emph{et~al.}(2017)Bai, Tsiatis, Lu, and Song}]{bai2017optimal}
Bai X, Tsiatis AA, Lu W, Song R (2017).
\newblock \enquote{Optimal treatment regimes for survival endpoints using a
  locally-efficient doubly-robust estimator from a classification perspective.}
\newblock \emph{Lifetime data analysis}, \textbf{23}, 585--604.

\bibitem[{Battocchi \emph{et~al.}(2019)Battocchi, Dillon, Hei, Lewis, Oka,
  Oprescu, and Syrgkanis}]{econml}
Battocchi K, Dillon E, Hei M, Lewis G, Oka P, Oprescu M, Syrgkanis V (2019).
\newblock \enquote{{EconML}: {A Python Package for ML-Based Heterogeneous
  Treatment Effects Estimation}.}
\newblock https://github.com/microsoft/EconML.
\newblock Version 0.14.0.

\bibitem[{Chakraborty \emph{et~al.}(2014)Chakraborty, Laber, and
  Zhao}]{chakraborty2014inference}
Chakraborty B, Laber EB, Zhao YQ (2014).
\newblock \enquote{Inference about the expected performance of a data-driven
  dynamic treatment regime.}
\newblock \emph{Clinical Trials}, \textbf{11}(4), 408--417.

\bibitem[{Chakraborty and Moodie(2013)}]{chakraborty2013statistical}
Chakraborty B, Moodie E (2013).
\newblock \emph{Statistical Methods for Dynamic Treatment Regimes}.
\newblock Springer-Verlag.

\bibitem[{Chen \emph{et~al.}(2020)Chen, Liu, Zeng, and Wang}]{cran:DTRlearn2}
Chen Y, Liu Y, Zeng D, Wang Y (2020).
\newblock \emph{\pkg{DTRlearn2}: Statistical Learning Methods for Optimizing
  Dynamic Treatment Regimes}.
\newblock \proglang{R} package version 1.1,
  \urlprefix\url{https://CRAN.R-project.org/package=DTRlearn2}.

\bibitem[{Chernozhukov \emph{et~al.}(2018)Chernozhukov, Chetverikov, Demirer,
  Duflo, Hansen, Newey, and Robins}]{chernozhukov2018double}
Chernozhukov V, Chetverikov D, Demirer M, Duflo E, Hansen C, Newey W, Robins J
  (2018).
\newblock \enquote{Double/Debiased Machine Learning For Treatment and
  Structural Parameters.}

\bibitem[{Cui \emph{et~al.}(2023)Cui, Kosorok, Sverdrup, Wager, and
  Zhu}]{cui2023heter}
Cui Y, Kosorok MR, Sverdrup E, Wager S, Zhu R (2023).
\newblock \enquote{Estimating heterogeneous treatment effects with
  right-censored data via causal survival forests.}
\newblock \emph{Journal of the Royal Statistical Society Series B: Statistical
  Methodology}, \textbf{85}(2), 179--211.

\bibitem[{D{\'\i}az \emph{et~al.}(2018)D{\'\i}az, Savenkov, and
  Ballman}]{diaz2018targeted}
D{\'\i}az I, Savenkov O, Ballman K (2018).
\newblock \enquote{Targeted learning ensembles for optimal individualized
  treatment rules with time-to-event outcomes.}
\newblock \emph{Biometrika}, \textbf{105}(3), 723--738.

\bibitem[{Ertefaie \emph{et~al.}(2012)Ertefaie, Almirall, Huang, Dziak, Wagner,
  and Murphy}]{procqlearn}
Ertefaie A, Almirall D, Huang L, Dziak JJ, Wagner A, Murphy S (2012).
\newblock \enquote{{SAS PROC QLEARN users' guide (Version 1.0.0)}.}
\newblock \urlprefix\url{http://methodology.psu.edu}.

\bibitem[{Goldberg and Kosorok(2012)}]{goldberg2012q}
Goldberg Y, Kosorok MR (2012).
\newblock \enquote{Q-learning With Censored Data.}
\newblock \emph{Annals of statistics}, \textbf{40}(1), 529.

\bibitem[{Hern{\'a}n and Robins(2020)}]{hernan2010causal}
Hern{\'a}n MA, Robins JM (2020).
\newblock \emph{Causal Inference: What If}.
\newblock Boca Raton: Chapman \& Hall/CRC.

\bibitem[{Hines \emph{et~al.}(2022)Hines, Dukes, Diaz-Ordaz, and
  Vansteelandt}]{hines2022demystifying}
Hines O, Dukes O, Diaz-Ordaz K, Vansteelandt S (2022).
\newblock \enquote{Demystifying Statistical Learning Based on Efficient
  Influence Functions.}
\newblock \emph{The American Statistician}, \textbf{76}(3), 292--304.

\bibitem[{Hirano and Porter(2012)}]{hirano2012impossibility}
Hirano K, Porter JR (2012).
\newblock \enquote{Impossibility Results for Nondifferentiable Functionals.}
\newblock \emph{Econometrica}, \textbf{80}(4), 1769--1790.

\bibitem[{Holloway \emph{et~al.}(2022)Holloway, Laber, Linn, Zhang, Davidian,
  and Tsiatis}]{cran:dyntxregime}
Holloway ST, Laber EB, Linn KA, Zhang B, Davidian M, Tsiatis AA (2022).
\newblock \emph{\pkg{DynTxRegime}: Methods for Estimating Optimal Dynamic
  Treatment Regimes}.
\newblock \proglang{R} package version 4.11,
  \urlprefix\url{https://CRAN.R-project.org/package=DynTxRegime}.

\bibitem[{Holst and Budtz-J{\o}rgensen(2013)}]{holst2013linear}
Holst KK, Budtz-J{\o}rgensen E (2013).
\newblock \enquote{Linear Latent Variable Models: the lava-package.}
\newblock \emph{Computational Statistics}, \textbf{28}(4), 1385--1452.

\bibitem[{Kennedy(2020)}]{kennedy2020towards}
Kennedy EH (2020).
\newblock \enquote{Towards Optimal Doubly Robust Estimation of Heterogeneous
  Causal Effects.}
\newblock \emph{arXiv preprint arXiv:2004.14497}.

\bibitem[{Kim \emph{et~al.}(2019{\natexlab{a}})Kim, Asher, Burkhauser, Mesite,
  and Leyva}]{DVN_AVW6KB_2019}
Kim JS, Asher CA, Burkhauser M, Mesite L, Leyva D (2019{\natexlab{a}}).
\newblock \enquote{{Replication Data for: Using a Sequential Multiple
  Assignment Randomized Trial (SMART) to Develop an Adaptive K–2 Literacy
  Intervention With Personalized Print Texts and App-Based Digital
  Activities}.}
\newblock \doi{10.7910/DVN/AVW6KB}.
\newblock \urlprefix\url{https://doi.org/10.7910/DVN/AVW6KB}.

\bibitem[{Kim \emph{et~al.}(2019{\natexlab{b}})Kim, Asher, Burkhauser, Mesite,
  and Leyva}]{kim2019using}
Kim JS, Asher CA, Burkhauser M, Mesite L, Leyva D (2019{\natexlab{b}}).
\newblock \enquote{Using a Sequential Multiple Assignment Randomized Trial
  (SMART) to Develop an Adaptive K--2 literacy intervention with personalized
  print texts and app-based digital activities.}
\newblock \emph{AERA Open}, \textbf{5}(3), 2332858419872701.

\bibitem[{K{\"u}nzel \emph{et~al.}(2019)K{\"u}nzel, Sekhon, Bickel, and
  Yu}]{kunzel2019metalearners}
K{\"u}nzel SR, Sekhon JS, Bickel PJ, Yu B (2019).
\newblock \enquote{Metalearners For Estimating Heterogeneous Treatment Effects
  Using Machine Learning.}
\newblock \emph{Proceedings of the national academy of sciences},
  \textbf{116}(10), 4156--4165.

\bibitem[{Lewis and Syrgkanis(2020)}]{lewis2020double}
Lewis G, Syrgkanis V (2020).
\newblock \enquote{Double/Debiased Machine Learning for Dynamic Treatment
  Effects via g-estimation.}
\newblock \emph{arXiv preprint arXiv:2002.07285}.

\bibitem[{Liang and Zeger(1986)}]{liang1986longitudinal}
Liang KY, Zeger SL (1986).
\newblock \enquote{Longitudinal Data Analysis Using Generalized Linear Models.}
\newblock \emph{Biometrika}, \textbf{73}(1), 13--22.

\bibitem[{Liu \emph{et~al.}(2018)Liu, Wang, Kosorok, Zhao, and
  Zeng}]{liu2018augmented}
Liu Y, Wang Y, Kosorok MR, Zhao Y, Zeng D (2018).
\newblock \enquote{Augmented Outcome-Weighted Learning For Estimating Optimal
  Dynamic Treatment Regimens.}
\newblock \emph{Statistics in medicine}, \textbf{37}(26), 3776--3788.

\bibitem[{Luedtke and Chambaz(2020)}]{luedtke2020performance}
Luedtke A, Chambaz A (2020).
\newblock \enquote{Performance Guarantees for Policy Learning.}
\newblock \emph{Annales de l'IHP Probabilites et statistiques}, \textbf{56}(3),
  2162.

\bibitem[{Luedtke and Chung(2023)}]{luedtke2023one}
Luedtke A, Chung I (2023).
\newblock \enquote{One-Step Estimation of Differentiable Hilbert-Valued
  Parameters.}
\newblock \emph{arXiv preprint arXiv:2303.16711}.

\bibitem[{Luedtke and van~der
  Laan(2016{\natexlab{a}})}]{luedtke2016statistical}
Luedtke AR, van~der Laan MJ (2016{\natexlab{a}}).
\newblock \enquote{Statistical Inference For the Mean Outcome Under a Possibly
  Non-unique Optimal Treatment Strategy.}
\newblock \emph{Annals of statistics}, \textbf{44}(2), 713.

\bibitem[{Luedtke and van~der Laan(2016{\natexlab{b}})}]{luedtke2016super}
Luedtke AR, van~der Laan MJ (2016{\natexlab{b}}).
\newblock \enquote{Super-learning of an Optimal Dynamic Treatment Rule.}
\newblock \emph{The international journal of biostatistics}, \textbf{12}(1),
  305--332.

\bibitem[{Nie and Wager(2021)}]{nie2021quasi}
Nie X, Wager S (2021).
\newblock \enquote{Quasi-oracle estimation of heterogeneous treatment effects.}
\newblock \emph{Biometrika}, \textbf{108}(2), 299--319.

\bibitem[{Nordland and Holst(2022)}]{cran:polle}
Nordland A, Holst KK (2022).
\newblock \emph{\pkg{polle}: Policy Learning}.
\newblock \proglang{R} package version 1.4,
  \urlprefix\url{https://CRAN.R-project.org/package=polle}.

\bibitem[{Petersen \emph{et~al.}(2012)Petersen, Porter, Gruber, Wang, and
  van~der Laan}]{petersen2012diagnosing}
Petersen ML, Porter KE, Gruber S, Wang Y, van~der Laan MJ (2012).
\newblock \enquote{Diagnosing and responding to violations in the positivity
  assumption.}
\newblock \emph{Statistical methods in medical research}, \textbf{21}(1),
  31--54.

\bibitem[{Polley \emph{et~al.}(2021)Polley, LeDell, Kennedy, and {van der
  Laan}}]{superlearner-package}
Polley E, LeDell E, Kennedy C, {van der Laan} M (2021).
\newblock \emph{\pkg{SuperLearner}: Super Learner Prediction}.
\newblock \proglang{R} package version 2.0-28,
  \urlprefix\url{https://CRAN.R-project.org/package=SuperLearner}.

\bibitem[{{\proglang{R} Core Team}(2022)}]{rcore}
{\proglang{R} Core Team} (2022).
\newblock \emph{\proglang{R}: A Language and Environment for Statistical
  Computing}.
\newblock R Foundation for Statistical Computing, Vienna, Austria.
\newblock \urlprefix\url{https://www.R-project.org/}.

\bibitem[{Robins(1986)}]{robins1986new}
Robins J (1986).
\newblock \enquote{A new approach to causal inference in mortality studies with
  a sustained exposure period—application to control of the healthy worker
  survivor effect.}
\newblock \emph{Mathematical modelling}, \textbf{7}(9-12), 1393--1512.

\bibitem[{Robins and Rotnitzky(2014)}]{robins2014discussion}
Robins J, Rotnitzky AG (2014).
\newblock \enquote{Discussion of “Dynamic treatment regimes: Technical
  challenges and applications”.}
\newblock \emph{Electron. J. Statist.}, \textbf{8}, 1273–--1289.
\newblock \doi{10.1214/14-EJS908}.

\bibitem[{Rubin(1974)}]{rubin1974estimating}
Rubin DB (1974).
\newblock \enquote{Estimating Causal Effects of Treatments in Randomized and
  Nonrandomized Studies.}
\newblock \emph{Journal of educational Psychology}, \textbf{66}(5), 688.

\bibitem[{Semenova and Chernozhukov(2021)}]{semenova2021debiased}
Semenova V, Chernozhukov V (2021).
\newblock \enquote{Debiased Machine Learning of Conditional Average Treatment
  Effects and Other Causal Functions.}
\newblock \emph{The Econometrics Journal}, \textbf{24}(2), 264--289.

\bibitem[{StataCorp(2021)}]{stata}
StataCorp L (2021).
\newblock \enquote{StataCorp. Stata Statistical Software: Release 17.}
\newblock \urlprefix\url{https://www.stata.com/}.

\bibitem[{Steingrimsson \emph{et~al.}(2016)Steingrimsson, Diao, Molinaro, and
  Strawderman}]{steingrimsson2016doubly}
Steingrimsson JA, Diao L, Molinaro AM, Strawderman RL (2016).
\newblock \enquote{Doubly robust survival trees.}
\newblock \emph{Statistics in medicine}, \textbf{35}(20), 3595--3612.

\bibitem[{Sverdrup \emph{et~al.}(2020)Sverdrup, Kanodia, Zhou, Athey, and
  Wager}]{sverdrup2020policytree}
Sverdrup E, Kanodia A, Zhou Z, Athey S, Wager S (2020).
\newblock \enquote{policytree: Policy Learning via Doubly Robust Empirical
  Welfare Maximization Over Trees.}
\newblock \emph{Journal of Open Source Software}, \textbf{5}(50), 2232.

\bibitem[{Sverdrup \emph{et~al.}(2022)Sverdrup, Kanodia, Zhou, Athey, and
  Wager}]{policytree}
Sverdrup E, Kanodia A, Zhou Z, Athey S, Wager S (2022).
\newblock \emph{\pkg{policytree}: Policy Learning via Doubly Robust Empirical
  Welfare Maximization over Trees}.
\newblock \proglang{R} package version 1.2.0,
  \urlprefix\url{https://CRAN.R-project.org/package=policytree}.

\bibitem[{Therneau(2023)}]{survival-package}
Therneau TM (2023).
\newblock \emph{A Package for Survival Analysis in \proglang{R}}.
\newblock \proglang{R} package version 3.5-3,
  \urlprefix\url{https://CRAN.R-project.org/package=survival}.

\bibitem[{Tibshirani \emph{et~al.}(2024)Tibshirani, Athey, Sverdrup, and
  Wager}]{grf2024}
Tibshirani J, Athey S, Sverdrup E, Wager S (2024).
\newblock \emph{grf: Generalized Random Forests}.
\newblock R package version 2.3.2,
  \urlprefix\url{https://CRAN.R-project.org/package=grf}.

\bibitem[{Tsiatis \emph{et~al.}(2019)Tsiatis, Davidian, Holloway, and
  Laber}]{tsiatis2019dynamic}
Tsiatis AA, Davidian M, Holloway ST, Laber EB (2019).
\newblock \emph{Dynamic treatment regimes: Statistical methods for precision
  medicine}.
\newblock Chapman and Hall/CRC.

\bibitem[{van~der Laan(2006)}]{van2006statistical}
van~der Laan MJ (2006).
\newblock \enquote{Statistical Inference For Variable Importance.}
\newblock \emph{The International Journal of Biostatistics}, \textbf{2}(1).

\bibitem[{van~der Laan and Luedtke(2014)}]{van2014targeted}
van~der Laan MJ, Luedtke AR (2014).
\newblock \enquote{Targeted Learning of an Optimal Dynamic Treatment, and
  Statistical Inference For Its Mean Outcome.}

\bibitem[{van~der Laan and Robins(2003)}]{van2003unified}
van~der Laan MJ, Robins JM (2003).
\newblock \emph{Unified Methods for Censored Longitudinal Data and Causality},
  volume~5.
\newblock Springer-Verlag.

\bibitem[{Xu \emph{et~al.}(2023)Xu, Ignatiadis, Sverdrup, Fleming, Wager, and
  Shah}]{xu2023treatment}
Xu Y, Ignatiadis N, Sverdrup E, Fleming S, Wager S, Shah N (2023).
\newblock \enquote{Treatment heterogeneity with survival outcomes.}
\newblock In \emph{Handbook of Matching and Weighting Adjustments for Causal
  Inference}, pp. 445--482. Chapman and Hall/CRC.

\bibitem[{Zhang \emph{et~al.}(2012)Zhang, Tsiatis, Davidian, Zhang, and
  Laber}]{zhang2012estimating}
Zhang B, Tsiatis AA, Davidian M, Zhang M, Laber E (2012).
\newblock \enquote{Estimating optimal treatment regimes from a classification
  perspective.}
\newblock \emph{Stat}, \textbf{1}(1), 103--114.

\bibitem[{Zhang \emph{et~al.}(2020)Zhang, Chen, Fu, He, Zhao, and
  Liu}]{zhang2020multicategory}
Zhang C, Chen J, Fu H, He X, Zhao YQ, Liu Y (2020).
\newblock \enquote{Multicategory Outcome Weighted Margin-based Learning For
  Estimating Individualized Treatment Rules.}
\newblock \emph{Statistica sinica}, \textbf{30}, 1857.

\bibitem[{Zhao \emph{et~al.}(2012)Zhao, Zeng, Rush, and
  Kosorok}]{zhao2012estimating}
Zhao Y, Zeng D, Rush AJ, Kosorok MR (2012).
\newblock \enquote{Estimating individualized treatment rules using outcome
  weighted learning.}
\newblock \emph{Journal of the American Statistical Association},
  \textbf{107}(499), 1106--1118.

\bibitem[{Zhou \emph{et~al.}(2018)Zhou, Athey, and Wager}]{zhou2018offline}
Zhou Z, Athey S, Wager S (2018).
\newblock \enquote{Offline Multi-action Policy Learning: Generalization and
  Optimization. arXiv preprint arXiv.}
\newblock \emph{arXiv preprint arXiv:1810.04778}.

\end{thebibliography}


\newpage

\begin{appendix}

\section{Binary $V$-optimal policy} \label{sec:v_optimal_policy}

Consider the two-stage case, $O = (S_1, A_1, S_2, A_2, U)$, where $A_1$ and $A_2$ are binary. Let $V_1$ be a function of $H_1 = S_1$ and let $(A_1, V_2)$ be a function of $H_2$. The $V$-optimal policy is defined as
\begin{align*}
d_0 = \arg \max_{d\in \mathcal{D}} E[U^d].
\end{align*}

The following theorem is a corrected version of Theorem 1 in \cite{van2014targeted}.
\begin{theorem} \label{theo:v_optimal_policy}
If $V_1$ is a function of $V_2$, then the $V$-optimal policy $d_{0}$ is given by
\begin{align*}
B_{0,2}(a_1, v_2) &= \E[U^{a_1,a_2 = 1} | V_2^{a_1} = v_2 ] - \E[U^{a_1,a_2 = 0} | V_2^{a_1} = v_2 ]\\
d_{0,2}(a_1, v_2) &= I\left\{B_{0,2}(a_1, v_2) > 0 \right\}\\
B_{0,1}(v_1) &= \E[U^{a_1 = 1,d_{0,2}} | V_1 = v_1 ] - \E[U^{a_1 = 0,d_{0,2}} | V_1 = v_1 ]\\
d_{0,1}(v_1) &= I\left\{B_{0,1}(v_1) > 0 \right\}.
\end{align*}
The above statement os also true if for all $a_1$ and $a_2$
\begin{align}
\E[U^{a_1,a_2}| V_1, V_2^{a_1}] = \E[U^{a_1,a_2}| V_2^{a_1}]. \label{eq:v_optimal_condition}
\end{align}
\end{theorem}
\begin{proof}
Let $V^a = (V_1, V_2^a)$. For any policy $d$
\begin{align*}
\E[U^d] &= \E\left[ \sum_{a_1, a_2} U^{a_1, a_2} I\{d_2(a_1, V_2^{a_1}) = a_2\} I\{d_1(V_1) = a_1\} \right]\\
&= \sum_{a_1} \E\left[ \left\{ \sum_{a_2} E \left( U^{a_1, a_2}\big| V_2^{a_1} \right) I\{d_2(a_1, V_2^{a_1}) = a_2\} \right\} I\{d_1(V_1) = a_1\}\right],
\end{align*}
where it is used that $V_1$ is a function of $V_2^{a_1}$ or that \eqref{eq:v_optimal_condition} holds. For any $a_1$ the inner sum is maximzed in $d_2$ by $d_{0,2}$, i.e., $\E[U^d] \leq E[U^{d_1, d_{0,2}}]$. Now,
\begin{align*}
\E[U^{d_1, d_{0,2}}] = \E\left[ \sum_{a_1} \E[U^{a_1, d_{0, 2}} | V_1] I\{d_1(V_1) = a_1\}\right],
\end{align*}
which is maximized for $d_1 = d_{0,1}$, i.e., $\E[U^d] \leq \E[U^{d_1, d_{0,2}}] \leq \E[U^{d_{0,1}, d_{0,2}}]$.
\end{proof}

\newpage

\section{Weighted classification loss function} \label{sec:class_loss}

We continue the setup from Appendix \ref{sec:v_optimal_policy}. At the second stage, define the doubly robust
blip score as
\begin{align*}
W_{2}(g, Q)(O) = Z_{2}(1, g, Q)(O)-Z_{2}(0, g, Q)(O),
\end{align*}
where $Z_{2}(a_{2}, g, Q)$ is the doubly robust score from \eqref{eq:Z_score_d_k}:
\begin{align*}
Z_{2}(a_{2}, g, Q)(O) = Q_{2}(H_2,a_{2}) + \frac{I\{A_2 = a_{2}\}}{g_{2}(H_2, A_2)} \left\{U - Q_{2}(H_{2}, A_{2}) \right\}.
\end{align*}
Now, define the loss function
\begin{align*}
L_{2}(d_2)(g_0, Q_{0})(O) = \lvert W_{2}(g_0, Q_{0})(O)\rvert I\Big\{d_2(A_1, V_2) \neq  I\{W_{2}(g_0, Q_{0})(O)> 0\}\Big\}.
\end{align*}
This is a valid loss function since the value loss function $\tilde L_{2}(d_2)(g_0)(O)$ from \eqref{eq:valuelossK} is a valid loss function and
\begin{align}
  -&\tilde L_{2}(d_2)(g, Q)(O) \\
 =&\,\, Q_{2}(H_2,d_{2}(H_{2})) + \frac{I\{A_2 = d_2(A_1, V_2)\}}{g_{2}(H_2, A_2)} \left\{U - Q_{2}(H_{2}, A_{2}) \right\} \nonumber \\
 =& \,\, d_2(A_1, V_2) W_{2}(g, Q)(O) + Q_{2}(H_2,0) + \frac{I\{A_2 = 0\}}{g_{2}(H_2, A_2)} \left\{ U - Q_{2}(H_{2},A_{2}) \right\}\nonumber \\
 =& \,\, I\{W_{2}(g, Q)(O) > 0\} \lvert W_{2}(g,Q)(O) \rvert +Q_{2}(H_2,0) + \frac{I\{A_2 = 0\}}{g_{2}(H_2, A_2)} \left\{ U - Q_{2}(H_{2},A_{2}) \right\} \label{eq:class_d_free} \\
 \quad&- \lvert W_{2}(g,Q)(O) \rvert I\Big \{d_2(A_1, V_2) \neq I\{W_{2}(g, Q)(O)> 0\} \Big\}.
\end{align}
Importantly, note that line \eqref{eq:class_d_free} does not depend on $d_{2}$. The last equality holds because for $d\in \{0,1\}$ and $W\in \mathbb{R}$
\begin{align*}
dW =& d|W|I\{W > 0\} - d|W|I\{W \leq 0\}\\
=& |W|I\{W > 0\} - |W|\Big((1-d)I\{W>0\} + d I\{W \leq 0\}\Big)\\
=& |W|I\{W > 0\} - |W|I\Big \{d \neq I\{W>0\} \Big\}.
\end{align*}
At the first stage, define the doubly robust blip score as
\begin{align*}
W_{1}(d_{2}, g, Q^{d_{2}}) =  Z_{1}([1,d_{2}], g, Q^{d_{2}})(O)-Z_{1}([0, d_{2}], g, Q^{d_{2}})(O).
\end{align*}
By similar calculations a valid loss function for $d_{0,1}$ is given by
\begin{align*}
 L_{1}(d_1)(d_{0,2},g_0, Q_{0}^{d_{0,2}})(O) = \lvert W_{1}(d_{0,2}, g_0, Q_{0}^{d_{0,2}})(O) \rvert I\Big\{d_1(V_1) \neq  I\{W_{1}(d_{0,2}, g_0, Q_{0}^{d_{0,2}})(O)> 0\}\Big \}.
\end{align*}
\end{appendix}


\end{document}